\documentclass[11pt, oneside]{amsart}
\usepackage[mathscr]{euscript}
\usepackage{amscd,amssymb,amsthm,verbatim}
\usepackage{enumerate}
\usepackage{bm}
\usepackage{mathrsfs, graphicx} 
\usepackage{color}
\usepackage{slashed}

 \usepackage{enumitem}
\numberwithin{equation}{section}

\theoremstyle{plain}
\newtheorem{theorem}{Theorem}[section]
\newtheorem{Theorem}{Theorem}
\newtheorem{proposition}[theorem]{Proposition}
\newtheorem{lemma}[theorem]{Lemma}
\newtheorem{corollary}[theorem]{Corollary}

\newtheorem{Lemma}[Theorem]{Lemma}

\newtheorem*{conjecture}{Conjecture}

\theoremstyle{definition}
\newtheorem{definition}[theorem]{Definition}

\newtheorem{remark}[theorem]{Remark}

\newtheorem{manualpropositioninner}{Proposition}
\newenvironment{manualproposition}[1]{%
  \renewcommand\themanualpropositioninner{#1}%
  \manualpropositioninner
}{\endmanualpropositioninner}

\newtheorem{manualtheoreminner}{Theorem}
\newenvironment{manualtheorem}[1]{%
  \renewcommand\themanualtheoreminner{#1}%
  \manualtheoreminner
}{\endmanualtheoreminner}

\newcommand{\n}{\mathbf n}
\newcommand{\loc}{\mathrm{loc}}

\renewcommand{\u}{\mathbf u}

\newcommand{\w}{\mathbf w}

\newcommand{\N}{\mathbf N}
\newcommand{\g}{\boldsymbol g}
\newcommand{\Y}{\mathbf Y}
\newcommand{\W}{\mathbf W}

\newcommand{\V}{\mathbf V}
\newcommand{\U}{\mathbf U}
\newcommand{\bnabla}{{\bm \nabla}}
\newcommand{\C}{\mathcal C}

\DeclareMathOperator{\Ric}{Ric}
\DeclareMathOperator{\Range}{Range}

\DeclareMathOperator{\Div}{div}
\DeclareMathOperator{\Det}{det}
\DeclareMathOperator{\adj}{adj}
\DeclareMathOperator{\tr}{tr} 

\DeclareMathOperator{\Int}{Int}
\DeclareMathOperator{\Span}{span}

\usepackage[tmargin=1in,bmargin=1in,lmargin=1.18in, rmargin=1.18in]{geometry}
\begin{document}
\title[Monotonicity and ADM  Mass Minimizers]{Monotonicity of Causal Killing Vectors and Geometry of ADM Mass Minimizers}
\author{Sven Hirsch}
\address{Columbia University, 2990 Broadway, New York NY 10027, USA}
\email{sven.hirsch@columbia.edu }
\author{Lan-Hsuan Huang}
\address{University of Connecticut, Storrs, CT 06269, USA}
\email{lan-hsuan.huang@uconn.edu}

\begin{abstract}
We address two problems concerning the ADM mass-minimizing initial data sets. First, we show that the equality case of the positive mass theorem embeds into a pp-wave spacetime. Second, we show that positive Bartnik mass minimizers embed into strongly stationary vacuum spacetimes, thereby confirming the Bartnik stationary vacuum conjecture. A key ingredient is a new monotonicity formula for the Lorentzian length of a causal Killing vector field, which, among other applications, yields a strong maximum principle for the length.
\end{abstract}

\maketitle

%\tableofcontents
%\addtocontents{toc}{\setcounter{tocdepth}{1}}  

\section{Introduction}

Monotonicity formulas play a fundamental role in  mathematical relativity, for example, as seen in the connections between weak cosmic censorship, Hawking’s area theorem, and the proofs of the Riemannian Penrose inequality~\cite{Huisken-Ilmanen:2001, Bray:2001, Agostiniani-Mantegazza-Mazzieri-Oronzio:2025}. For a spacetime  admitting a causal Killing vector field~$\Y$, we derive a new monotonicity formula along the level sets of the length function of $\Y$ on initial data sets transverse to $\Y$.

 We let $n \ge 3$ and $(\N, \g)$ be an $(n+1)$-dimensional spacetime. We say that $\Y$ is \emph{causal} if $\Y$ is timelike or null, or equivalently,  $\eta := -\g(\Y, \Y) \ge 0$. The  monotonicity formula states:

\begin{Theorem}[Monotonicity formula]\label{TH:monotonicity}
Let $(\N, \g)$ be a spacetime admitting a future-directed causal Killing vector field $\Y$ satisfying $\Ric_{\g}(\Y, \Y) \le 0$. Let $U$ be spacelike hypersurface in $\N$ with the future-directed unit normal $\n$. Denote by $(f, X)$ the lapse-shift pair of $\Y$ along $U$; namely, $\Y = f \n +X$.    Suppose the level  sets  $\Sigma_t = \{ x\in  U: \eta (x) = t\}$ are compact without boundary for regular values of $t$, and define, for $t>0$, 
\begin{align*}
 M(t) &:= t^{-1}\int_{\Sigma_t}\tfrac{f}{|\nabla \eta|}  \big(|\nabla \eta|^2 - \big(\nabla \eta \cdot \tfrac{ X}{f}\big)^2 \big)\, d\sigma
\end{align*}
where $\nabla$ is the gradient with respect to the induced metric $g$ on $U$. Then $M(t)$ is non-increasing in $t$. 

Moreover, if $M(t)$ is constant for $t\in (a, b)$, then   $\Ric_{\g}(\Y, \Y)=0$, the $g$-dual one-form of $X$ is closed, and $\nabla \eta \cdot X=0$ on $\Omega_b \setminus \Omega_a$ where $\Omega_a := \{ x\in U: \eta(x)<a\}$.

\end{Theorem}

We refer to Theorem~\ref{th:monotonicity} for a more detailed statement. The monotonicity formula yields the following applications.

An asymptotically flat spacetime $(\N, \g)$ with a Killing vector field $\Y$ is called \emph{stationary} if $\Y$ is timelike in the exterior of an asymptotically flat initial data set $(M, g, k)$, and \emph{uniformly stationary} if there exists a constant $c>0$ such that $\g(\Y,\Y)<-c$ everywhere on the exterior of $M$. A uniformly timelike Killing vector provides a natural time coordinate in the asymptotically flat region, which is fundamental for results such as the uniqueness of stationary black holes and the existence of maximal hypersurfaces; see, e.g.,~\cite{Chrusciel-Costa-Heusler:2012, Chrusciel-Wald:1994}. R.~Beig and P.~Chru\'sciel~\cite{Beig-Chrusciel:1996} and Chru\'sciel and D.~Maerten~\cite{Chrusciel-Maerten:2006} proved that stationarity \emph{implies} uniform stationarity under certain global assumptions on initial data sets. Specifically, they showed that an asymptotically null Killing vector corresponds to null ADM energy-momentum and imposed global conditions to rule out non-Minkowski spacetime. However, the abundance of pp-wave examples with a global null Killing vector and null ADM energy–momentum shows that these assumptions are restrictive.

We find local conditions near infinity that imply uniform stationarity.

\begin{Theorem}[Uniform stationarity]\label{Th:uni}
Let $(\N, \g)$ be a spacetime admitting a Killing vector field $\Y$, and let $(M, g, k)$ be  an asymptotically flat initial data set in $\N$, where $M$ may have compact boundary. If, along $M$, the Einstein tensor $G = O^{0, \alpha}(|x|^{-n-\epsilon})$ for some $\epsilon>0$, $\Ric_{\g}(\Y, \Y) \le 0$, and $\Y$ is timelike, then $\Y$ is uniformly timelike on $M$.
\end{Theorem}

We remark that $G = O^{0,\alpha}(|x|^{-n-\epsilon})$ in the above statement means that $G(\n,\n)$, $G(\n,\partial_i)$, and $G(\partial_i,\partial_j)$ are each $O^{0,\alpha}(|x|^{-n-\epsilon})$ on $M$, with respect to a unit normal $\n$ to $\N$ and the asymptotically flat coordinate vectors $\partial_i$.

A key step in the monotonicity formula is to restrict an inhomogeneous wave equation satisfied by $\eta$ to a spacelike hypersurface. This yields a divergence-type elliptic equation for~$\eta$, which is degenerate precisely on the null set of $\Y$. Analyzing this null set leads to the following regularity result.
\begin{Lemma}[Regularity of the null set]\label{Le:smooth}
Let $(\N, \g)$ be a spacetime admitting a causal Killing vector $\Y$ satisfying $\Ric_{\g} (\Y, \Y)\le0$. 
Define the null set $\mathbf Z = \{ x\in \N: \eta(x)=0\}$. Let $U\subset \N$ be a smooth spacelike hypersurface. Then $\partial \mathbf Z\cap \Int U$ is a hypersurface in $U$ as smooth as~$\g$. 
\end{Lemma}

In special circumstances (e.g., stationary spacetimes), an event horizon is a Killing horizon (where the Killing vector becomes null). Consequently, the above result may be compared with the regularity of cross-sections of event horizons established in \cite{Chrusciel-Delay-Galloway-Howard:2001, Chrusciel-Costa:2008, Hounnonkpe-Minguzzi:2025} in connection with the Hawking rigidity theorem. Their results, however, are stronger, since event horizons are global geometric objects, whereas a Killing horizon, as well as our argument, is local.

Using the localized monotonicity formula (Proposition~\ref{pr:mono}), together with the preceding regularity result, we show that the length of a causal Killing vector satisfies a strong maximum principle.
\begin{Theorem}[Strong maximum principle]\label{Th:strong}
Let $(\N,\g)$ be a spacetime admitting a causal Killing vector field $\Y$ with $\Ric_{\g}(\Y,\Y)\le 0$, and let $(U, g, k)$ be an initial data set in $\N$ with smooth boundary $\partial U$. Denote $\eta = -\g(\Y, \Y)\ge 0$. Then:  
\begin{enumerate}
\item Suppose there exists a point $p\in \partial U$ such that $\eta(x) > \eta(p)$ for all $x\in \Int U$. If $\eta(p)>0$, then $\nabla \eta(p)\neq 0$. If $\eta=0$ on an open subset $\Sigma\subset \partial U$ containing $p$, then $\nabla \eta \neq 0$ on a dense subset of $\Sigma$.  
\item If $\eta$ attains a local minimum at an interior point of $U$, then $\eta$ is constant on $U$.  
\end{enumerate}
\end{Theorem}

The above result implies that, in particular, if $\Y$ is null at an interior point of $\N$, then $\Y$ is null everywhere.  Therefore, a Killing vector field must change from timelike to spacelike when crossing a Killing horizon, as illustrated in explicit examples such as Kerr spacetimes.

\medskip

Our primary motivation for obtaining the above results is to settle the equality case of the positive mass theorem and the Bartnik stationary vacuum conjecture, as discussed below.

The development of the positive mass theorem for asymptotically flat initial data sets has a long history; see, for example, \cite{Schoen-Yau:1979-pmt1, Witten:1981, Schoen-Yau:1981-pmt2, Yip:1987, Eichmair:2013, Beig-Chrusciel:1996, Chrusciel-Maerten:2006, Huang-Lee:2020, HKK:2022, HirschZhang:2023, Huang-Lee:2025, Huang-Lee:2024, Hirsch-Zhang:2024} and references therein. Regarding the equality case of the positive mass theorem, it has been shown that a complete $n$-dimensional asymptotically flat initial data set satisfying the dominant energy condition with  ADM energy-momentum $E = |P|$ must embed in Minkowski spacetime for $n = 3, 4$, and  for $n>4$ provided additional asymptotically flat fall-off conditions are assumed~\cite{Beig-Chrusciel:1996, Chrusciel-Maerten:2006, Huang-Lee:2020, HirschZhang:2023, Huang-Lee:2025}.  Here, we say that an initial data set $(U, g, k)$  \emph{embeds} in a spacetime $(\mathbf{N}, \g)$ of one higher dimension if $(U, g)$ isometrically embeds into $(\mathbf{N}, \g)$ with $k$ as its second fundamental form.

Without these additional fall-off rates, the conjectured characterization in Minkowski spacetime is no longer valid. D. Lee and the second-named author~\cite{Huang-Lee:2024} constructed counterexamples in pp-waves for $n>8$. Y. Zhang and the first-named author~\cite{Hirsch-Zhang:2024} improved this to $n>4$, and further showed that if the initial data set is spin, it must embed in a pp-wave (see also \cite{Witten:1981, Beig-Chrusciel:1996, Chrusciel-Maerten:2006}).  We remove the spin assumption.  

%It is natural to conjecture that this characterization still holds without the spin assumption.

%D. Lee and the second-named author~\cite{Huang-Lee:2024} employed a variational argument to show that an ADM mass–minimizing initial data set must embed in a spacetime $(\N, \g)$ with a Killing vector field $\Y$ satisfying $\Ric_{\g} = \zeta  \Y \otimes \Y$, where $\zeta$ vanishes whenever $\Y$ is not null. In particular, such a spacetime satisfies the Killing vector energy condition $\Ric_{\g}(\Y, \Y)=0$, as required in previous results. 

%Building on the prior result, we refine the variational argument to show that $\Y$ is causal, and by applying the causal characterization above, we establish the equality case without assuming spin.
\begin{Theorem}[Equality case of the positive mass theorem]\label{Th:equality}
Let $(M, g, k)$ be an $n$-dimensional, complete, boundaryless, asymptotically flat initial data set with $(g, k)\in \C^5_{\mathrm{loc}}(M)\times \C^4_{\mathrm{loc}}(M)$. Suppose it satisfies the dominant energy condition, has ADM energy–momentum $(E, P)$ with $E = |P|$, and that the positive mass inequality holds in a neighborhood of $(g, k)$. Then $(M, g, k)$ embeds in a pp-wave\footnote{We say that $(\N, \g, \Y)$ is a pp-wave if the spacetime $(\N, \g)$  admits a parallel null vector field $\Y$ so that each integral hypersurface of $\Y^\perp$, which is a null hypersurface, has flat induced metric, meaning its curvature tensor vanishes identically.   (The term ``pp-wave'' is short for plane-fronted waves with parallel rays.)} satisfying the spacetime dominant energy condition.
\end{Theorem}

We say that the \emph{positive mass inequality holds near $(g, k)$} if there exists an open ball around $(g, k)$ in $\C^{2,\alpha}_{-q}(M)\times \C^{1,\alpha}_{-1-q}(M)$ such that every asymptotically flat initial data set $(\bar g, \bar k)$ in it, satisfying the dominant energy condition, has ADM energy–momentum $\overline E \ge |\overline P|$.

To reach the final characterization by pp-waves in Theorem \ref{Th:equality}, we derive the following result, which may be of independent interest.
\begin{Theorem}\label{Th:pp}
Let $(\N, \g)$  be a spacetime admitting a null Killing vector field~$\Y$  and an asymptotically flat initial data set $(M, g, k)$ where $M$ is complete without boundary. Assume  $\Ric_{\g}(\Y, \Y) \le 0$ on $\N$ and the  lapse-shift pair $(f, X)$ of $\Y$ is asymptotically translational on~$M$. Then $\Ric_{\g}(\Y, \Y)\equiv 0$, and the following holds:
\begin{enumerate}
\item Assume $\Ric_{\g}(\Y, \n) = 0$ along $M$, where $\n$ is the unit normal to $M$, then $\Y$ is parallel on the domain of dependence of $M$ in $(\N, \g)$. 
\item Assume $\Y$ is parallel and $\Ric_{\g}(\V, \W) = 0$ for all spacelike vector fields $\V, \W \in \Y^\perp$ where $\Y^\perp$ denotes the orthogonal distribution to $\Y$, then $(\N, \g, \Y)$ is a pp-wave.%Ricci-flat. Moreover, if either $n=3,4$ or $n\ge 5$ and $S$ contains a complete, boundaryless, asymptotically flat  $(n-1)$-dimensional spacelike submanifold, then  the induced metric on $S$ is flat. 
\end{enumerate}
\end{Theorem}

Our proof of Theorem~\ref{Th:equality} extends to initial data sets with compact boundary, particularly in the context of Bartnik's quasi-local mass program.  For an initial data set $(U, g, k)$, we define the \emph{Bartnik boundary data} on $\partial U$ as
\[
	 B(g, k) = \big(g^\intercal, H_g, k(\nu_g)^\intercal, \tr_{g^\intercal} k \big),
\]
where $g^\intercal$ denotes the induced metric, $H_g$ denotes the mean curvature with respect a specified unit normal $\nu_g$, $k(\nu_g)^\intercal$ denotes the tangential components of $k(\nu_g, \cdot)$, and $\tr_{g^\intercal} k$ is the tangential trace.

Let $(\Omega, g_0, k_0)$ be a compact initial data set with nonempty smooth boundary. The \emph{Bartnik mass} is defined as
\[
m_B(\Omega, g_0, k_0) =\inf\big \{ m_{\mathrm {ADM}} (M, g, k): \mbox{admissible extension } (M, g, k)\big\}.
\]
Here, an ``extension'' refers to an asymptotically flat manifold $M$ with boundary $\partial M$ diffeomorphic to $\partial \Omega$. Among other conditions, an admissible extension must match the Bartnik boundary data $B(g, k) = B(g_0, k_0)$. See Definition~\ref{definition:admissible} for the full definition.

An admissible extension $(M, g, k)$ whose ADM mass realizes the Bartnik mass is called a \emph{Bartnik mass minimizer} for $(\Omega, g_0, k_0)$.  Bartnik conjectured that minimizers should be vacuum and stationary (see \cite[p. 2348]{Bartnik:1989}, \cite[Conjecture~2]{Bartnik:1997}, and \cite[p. 236]{Bartnik:2002}). (We remark that the original conjecture was for three dimensions, but it can naturally be posed in general dimensions.)

\begin{conjecture}[Bartnik's Stationary Vacuum Conjecture, 1989]
Let $(\Omega, g_0, k_0)$ be a compact initial data set with nonempty smooth boundary, satisfying the dominant energy condition, and suppose that $(M, g, k)$ is a Bartnik mass minimizer for $(\Omega, g_0, k_0)$. Then $(\Int M, g, k)$ embeds in a vacuum spacetime admitting a  timelike Killing vector field. Or in other words, $(M, g, k)$ embeds in a vacuum and strongly stationary spacetime. 
\end{conjecture}

There has been important partial progress toward the above conjecture. The so-called Riemannian case ($k \equiv 0$) was studied by J.~Corvino~\cite{Corvino:2000} and M. Anderson and J.~Jauregui~\cite{Anderson-Jauregui:2019}. Z.~An~\cite{An:2020} confirmed the conjecture for vacuum initial data in $n=3$. D.~Lee and the second-named author~\cite{Huang-Lee:2024} showed that the minimizer $(M, g, k)$ embeds in a certain spacetime $(\mathbf{N}, \g)$ with a Killing vector field $\Y$. On the other hand,  as discussed in \cite[Example 7]{Huang-Lee:2024}, the pp-wave counterexamples mentioned above also disprove the conjecture for \emph{zero} Bartnik mass in $n>4$.

We confirm the conjecture for \emph{positive} Bartnik mass, as a consequence of the following more general result. This result is more general than Bartnik’s original requirement, as $(M, g, k)$ need not admit a ``fill-in'' $(\Omega, g_0, k_0)$ satisfying the dominant energy condition.

\begin{Theorem}\label{Th:Bartnik}
Let $(M, g, k)$ be an asymptotically flat initial data set with compact boundary and satisfies the dominant energy condition. Assume the ADM energy-momentum $E> |P|$ and $(g, k)\in\C^5_{\mathrm{loc}}(\Int M)\times\C^4_{\mathrm{loc}}(\Int M)$. If  $(g, k)$ minimizes the ADM mass among nearby initial data sets with the same Bartnik boundary data, then $(\Int M, g, k)$ embeds in a vacuum and strongly stationary spacetime. 
\end{Theorem}

We say that $(g, k)$ \emph{minimizes the ADM mass among nearby initial data sets with the same Bartnik boundary data} if there is an open ball around $(g, k)$ in $\C^{2,\alpha}_{-q}(M)\times \C^{1,\alpha}_{-1-q}(M)$ such that every asymptotically flat initial data set $(\bar g, \bar k)$ in it, satisfying the dominant energy condition and $B(\bar g, \bar k) = B(g, k)$, obeys $\bar E \ge |\bar P|$ and ${\overline E}^2 - |\overline P|^2 \ge E^2 - |P|^2$.

An immediate application of Theorem~\ref{Th:Bartnik} is that, in Kerr spacetimes of nonzero angular momentum, the presence of the ergoregion, where the asymptotically translational Killing vector field becomes spacelike, implies that an asymptotically flat initial data set in a Kerr spacetime whose boundary lies on the event horizon cannot minimize the ADM mass.

\bigskip 

\noindent \textbf{Structure of the paper:} 
We give an overview of how Theorems~\ref{Th:uni} and~\ref{Th:strong} are used to establish Theorems~\ref{Th:equality} and~\ref{Th:Bartnik}. Let $(M, g, k)$ be the ADM mass–minimizing initial data set from Theorems~\ref{Th:equality} and~\ref{Th:Bartnik}. Using prior results and arguments in \cite{Huang-Lee:2024}, one can show that  $(M, g, k)$ embeds into the Killing development $(\N, \g, \Y)$ satisfying $\Ric_{\g} (\Y, \Y)=0$. (This is the only place where the $\C^5_{\mathrm{loc}}\times \C^4_{\mathrm{loc}}$ regularity assumption is needed; see Remark~\ref{re:reg}.) We further show, via a variational argument, that $\Y$ must be causal;  otherwise one can construct nearby initial data sets with smaller ADM mass. Hence, Theorems~\ref{Th:uni} and~\ref{Th:strong} imply that $\Y$ is null everywhere if the ADM mass is zero, and strongly, uniformly timelike everywhere if the ADM mass is positive. In the former case, we then apply Theorem~\ref{Th:pp} to further characterize such spacetimes as pp-waves.

The paper is organized as follows. Section~\ref{se:mono} proves the monotonicity formula. The uniform stationarity result for asymptotically flat spacetimes (Theorem~\ref{Th:uni}) is established in Section~\ref{se:st}. Section~\ref{se:max} proves the regularity of the null set (Lemma~\ref{Le:smooth}) and the strong maximum principle (Theorem~\ref{Th:strong}). Section~\ref{se:spacelike} shows that initial data sets with a lapse–shift pair that is spacelike somewhere do not minimize the ADM mass. Theorems~\ref{Th:equality} and~\ref{Th:Bartnik} are proved in Section~\ref{se:min}, and Theorem~\ref{Th:pp} is proved in Section~\ref{S:vector}. Appendix~\ref{sec:af} collects basic definitions and facts for asymptotically flat initial data sets, Appendix~\ref{se:de} reviews classical results for degenerate elliptic PDEs, and Appendix~\ref{se:sur} presents a local surjectivity result for the modified constraint operator with Bartnik boundary data, generalizing some work of Anderson–Jauregui~\cite{Anderson-Jauregui:2019} and An~\cite{An:2020} by a different argument. 
\medskip

\noindent \textbf{Acknowledgements:}  SH was supported in part by the National Science Foundation under Grant No. DMS-1926686, and by the IAS School of Mathematics. Moreover, he is grateful for the hospitality of the Department of Mathematics of the University of Connecticut where this project was initiated. LH was partially supported by the National Science Foundation under Grant No. DMS-2304966 and a Simons Fellowship.

\section{Monotonicity Formulas}\label{se:mono}

We let $n\ge 3$ and $(\N,\g)$ be an $(n+1)$-dimensional connected, time-oriented Lorentzian manifold, referred to as a \emph{spacetime}. Let $U$ be a spacelike  hypersurface in $\N$ with the future-directed unit normal $\n$. Let $g$ be the induced metric and $k = (\bnabla \n)^\intercal$  the second fundamental form.  Along $U$, the vector field $\Y$  decomposes as $\Y = f\n + X$, where $f$ is a scalar function on $U$ and $X$ is a vector field tangent to $U$. The pair $(f, X)$ is referred to as the \emph{lapse-shift pair} of $\Y$. We say a spacetime is asymptotically flat if it admits an asymptotically flat initial data set. 

 We use boldface to denote geometric operators associated with the spacetime metric $\g$. For example, $\bnabla$ and $\bm{\Gamma}$ denote the covariant derivative and Christoffel symbols with respect to $\g$, while $\nabla$ and $\Gamma$ refer to those associated with the induced spatial metric $g$. 
 
 Let  $e_0, e_1, \dots, e_n$ be a local frame of $\N$ with $e_0=\n$  the future-directed unit normal of $U$, and $e_1, \dots, e_n$  tangent to $U$.  Unless otherwise indicated, Greek indices $\alpha, \beta, \dots$ run over all space and time directions ${0, 1, \dots, n}$, while Roman indices $i, j, k, \ell$  run over ${1, 2, \dots, n}$. We denote covariant derivatives by $_{|\alpha}$ for $\g$ and by $_{;i}$ for  $g$, and ordinary derivatives by $_{,\alpha}$. We use the Einstein summation convention, summing over repeated indices, unless otherwise indicated.

We say that a vector field $\Y$ on $\N$ is \emph{Killing}  if $\Y$ satisfies $\mathscr L_\Y \g=0$, or equivalently $\Y_{\alpha|\beta} + \Y_{\beta|\alpha}=0$, and it is not identically zero.  It is straightforward to verify that a Killing vector field satisfies
\begin{align}\label{eq:DDY}
	\Y_{\alpha| \beta \gamma} = \mathbf R_{\beta\alpha \gamma \delta} \Y^\delta
\end{align}
where  the Riemann curvature tensor of $\g$ is defined by 
\[
\mathbf R_{\beta\alpha \gamma \delta}=\mathrm{Rm}_{\g} ( e_\beta, e_\alpha, e_\gamma, e_\delta): = \g\big( (\bnabla_{e_\beta} \bnabla_{e_\alpha} -\bnabla_{e_\alpha} \bnabla_{e_\beta} - \bnabla_{[e_\beta, e_\alpha]}) e_\gamma, e_\delta\big)
\]
and the Ricci curvature by $\Ric_{\g}(e_{\alpha}, e_{\gamma} )=  \g^{\beta \delta} \mathbf R_{\beta\alpha \gamma \delta}$.

For a Killing vector field $\Y$, the Bochner identity implies that
\begin{align}\label{eq:bochner}
	\Box_{\g} \g(\Y, \Y)  = 2 \g(\bnabla \Y, \bnabla \Y) - 2\Ric_{\g}(\Y, \Y),
\end{align}
where $\Box_{\g} := \Div_{\g} \bnabla$ is the wave operator and 
\begin{align}\label{eq:grad}
\g (\bm \nabla \Y, \bm \nabla \Y) :=   \g^{\alpha \gamma} \g^{\beta \delta} \Y_{\alpha|\beta} \Y_{\gamma|\delta} = -  \g^{\alpha \gamma} \g^{\beta \delta} \Y_{\alpha|\beta} \Y_{\delta|\gamma}=- \Y^\gamma_{\;\;|\beta}  \Y_{\;\;| \gamma}^\beta.
\end{align}

We would like to express \eqref{eq:bochner} in terms of initial data. To this end, it is convenient to choose coordinates adapted to the Killing vector field $\Y$.  Let $(f, X)$ be the lapse-shift of $\Y$, and suppose $f \neq 0$. Then there exists a tubular neighborhood $\U$ of $U$ in $\N$ such that
\begin{align}\label{eq:g}
\g = -f^2  du^2 + g_{ij} (dx^i + X^i du)(dx^j + X^j du),
\end{align}
where $u$ is the coordinate along the flow lines of $\Y$, so that $\Y = \frac{\partial}{\partial u}$, $(x^1, \dots, x^n)$ is a local coordinate chart on $U$, and $f, X, g_{ij}$ are independent of the coordinate~$u$. 

For later use in Section~\ref{S:vector}, we record the following identity, as shown in \cite[p.~794]{Huang-Lee:2024}:
\begin{align}\label{eq:n}
\bnabla_{\n} \n = f^{-1} \nabla f.
\end{align} We include the derivation for completeness. Using $\n = -f \bnabla u$, for any vector $e_i$ tangent to~$U$, 
\begin{align*}
	\g( \bnabla_{\n} \n, e_i) &= - \g(\bnabla_{\n} (f \bnabla u), e_i) = - \g(\bnabla_{f\n} \bnabla u, e_i)= - \bnabla^2 u(f\n, e_i) \\
	&= - \bnabla_{e_i} \bnabla_{f\n} u + \bnabla_{\bnabla_{e_i} (f\n)} u = e_i(f) \bnabla_{\n} u + f\bnabla_{\bnabla_{e_i} \n} u \\
	&= f^{-1} e_i(f)
\end{align*}
where we  used $\bnabla_{f\n} u=1$ and  that $\bnabla_{e_i} \n$ is tangent to $U$.

We set 
\[
\eta = -\g(\Y, \Y) = f^2 - |X|_g^2 \quad \mbox{ and } \quad \hat X = \tfrac{X}{f}.
\] 
In these coordinates,  the components of $\g$,  $\g^{-1}$ are given by 
\begin{align*}
	\g_{uu} &= - \eta , \quad \g_{ui} = g_{ij} X^j =: X_i, \quad  \g_{ij} = g_{ij},\\ \g^{uu} &= -\tfrac{1}{f^2}, \quad \g^{ui} =  \tfrac{X^i}{f^2}, \quad \g^{ij} =  g^{ij} - \hat X^i \hat X^j.
\end{align*}
The Christoffel symbols are given by 
\begin{align*}
	{\mathbf \Gamma}_{uu}^u &=\tfrac{1}{2f^2}\nabla \eta \cdot X\\
	{\mathbf \Gamma}_{uu}^k &=  \tfrac{1}{2} \g^{kj} \eta_{,j}\\
	\mathbf \Gamma_{ju}^u &= \tfrac{1}{2f^2} \eta_{,j} + \tfrac{1}{2f^2} X^k(X_{k;j} - X_{j;k})\\
		{\mathbf \Gamma}_{ui}^j &= \tfrac{1}{2} \g^{jk} (X_{k;i} - X_{i;k} ) - \tfrac{1}{2f^2} \eta_{,i} X^j \\
		{\mathbf \Gamma}_{ij}^k &=   \tfrac{1}{2f^2} X^k(X_{i;j} + X_{j;i})  + \Gamma_{ij}^k.
\end{align*}

\begin{lemma}\label{le:div}
Let  $(\N, \g)$ be a spacetime admitting a   Killing vector field $\Y$ and an initial data set $(U, g, k)$. Let $(f, X)$ be the lapse-shift pair of $\Y$ and suppose $f\neq 0$.   If $\W = W^0 \frac{\partial}{\partial u} + W$ along $U$ where $W^0$ is independent of $u$ and $W = W^j \frac{\partial}{\partial x^j}$ is tangent to $U$. Then 
\begin{align}\label{eq:divw}
	\Div_{\g} \W = \tfrac{1}{f} \Div_g (  f W).
\end{align}
\end{lemma}
\begin{proof}
We compute 
\begin{align*}
\Div_{\g} \mathbf W &=W^\alpha_{,\alpha} + \bm \Gamma^\alpha_{\alpha \gamma}  W^\gamma\\
&=
 W^j_{,j} + \bm \Gamma_{\alpha u}^\alpha W^0+  \bm \Gamma_{\alpha j}^\alpha W^j\\
& =W^j_{,j}  + \Big( \tfrac{1}{2f^2} \eta_j  +  \tfrac{1}{f^2} X^kX_{k;j} + \Gamma_{kj}^k\Big) W^j\\
	&=  \Div_g W + \Big(  \tfrac{1}{2f^2} \eta_j + \tfrac{1}{f^2} X^kX_{k;j} \Big)  W^j\\
	&= \tfrac{1}{f} \Div_g ( fW).
\end{align*}
\end{proof}

\begin{lemma}\label{le:elliptic1}
Let  $(\N, \g)$ be a spacetime admitting a   Killing vector field $\Y$ and an initial data set $(U, g, k)$. Let $(f, X)$ be the lapse-shift pair of $\Y$ and suppose $f\neq 0$. %\begin{enumerate}
%\item If $\W = W^0 \frac{\partial}{\partial u} + W$ along $U$ where $W^0$ is independent of $u$ and $W = W^j \frac{\partial}{\partial x^j}$ is tangent to $U$. Then 
%\begin{align}\label{eq:divw}
%	\Div_{\g} \W = \tfrac{1}{f} \Div_g (  f W).
%\end{align}
Suppose $w$ is a scalar function on $\N$  that is independent of~$u$. Then  the following identity holds on $U$:
\begin{align*}%\label{eq:elliptic}
\Box_{\g} w&= \tfrac{1}{f} \Div_g \Big(  f \big(\nabla w -( \nabla w \cdot \hat X)\hat X\big)\Big),
\end{align*}
where the dot denotes the $g$-inner product. 
\end{lemma}
\begin{proof}

Computing $\bm \nabla w$ in local coordinates gives
\begin{align*}
	 \bm \nabla w = \g^{\alpha \beta} \tfrac{\partial w}{\partial x^\alpha} \tfrac{\partial}{\partial x^\beta}=\tfrac{X^i}{f^2} \tfrac{\partial w}{\partial x^i} \tfrac{\partial}{\partial u} + \g^{ij} \tfrac{\partial w}{\partial x^i} \tfrac{\partial}{\partial x^j} = W^0\tfrac{\partial}{\partial u}+ W
\end{align*}
where  $W^0$ is independent of $u$, and 
$W =  \g^{ij} \tfrac{\partial w}{\partial x^i} \tfrac{\partial}{\partial x^j}  = \nabla w - (\nabla w \cdot \hat X)\hat X$ is tangent to $U$. The desired identity follows Lemma~\ref{le:div}. 
\end{proof}

\begin{proposition}\label{pr:elliptic2}
Let  $(\N, \g)$ be a spacetime admitting a  Killing vector field $\Y$ and an initial data set $(U, g, k)$. Let $(f, X)$ be the lapse-shift pair of $\Y$ and suppose $f\neq 0$.  Then 
\begin{enumerate}[leftmargin=*]
\item  This identity holds on $U$:
\begin{align} \label{eq:nablaY}
\begin{split}
\g (\bm \nabla \Y , \bm \nabla \Y)&=- \tfrac{1}{2f^2} |\nabla \eta|^2 -   \tfrac{1}{f} (\nabla \eta)^i\hat X^k (X_{k;i} - X_{i;k}) \\
&\quad +  \tfrac{1}{4} \g^{ik} (X_{k;j} - X_{j;k}) \g^{j\ell} (X_{ i; \ell} - X_{\ell;i}).
\end{split}
\end{align}
Moreover, if $\Y$ is null, then the last term $\tfrac{1}{4} \g^{ik} (X_{k;j} - X_{j;k}) \g^{j\ell} (X_{ i; \ell} - X_{\ell;i})$  vanishes if and only if  the orthogonal distribution $X^\perp$ is involutive. If $\Y$ is timelike, then the last term vanishes if and only if the $g$-dual one-form of $X$ is closed.  
\item If $\eta\neq 0$,  then
\begin{align} \label{eq:nablaY2}
\begin{split}
	&\g (\bm \nabla \Y , \bm \nabla \Y)
	=  -\tfrac{1}{2\eta} \big(|\nabla\eta|^2 - (\nabla \eta\cdot \hat X)^2\big)\\
	&\qquad + \tfrac{1}{2}\Big(- \tfrac{| X|}{f^2 (1-|\hat X|^2)^{\frac{1}{2}} } \nabla \eta^\perp +(1-|\hat X|^2)^{\frac{1}{2}}\sum_{a=1}^{n-1} (X_{n;a} - X_{a;n})e_a \Big)^2\\
	&\qquad  +  \tfrac{1}{4}\sum_{a,b=1}^{n-1}(X_{a;b}- X_{b;a})^2, 
\end{split}
\end{align}
where  the terms in the last two lines are expressed with respect to a local orthonormal frame $\{ e_1, \dots, e_{n-1}, e_n\}$ on $(U, g)$ with $e_n = \frac{X}{|X|_g}$ if $X\neq 0$, and the notation $\nabla \eta^\perp$ denotes the component of $\nabla \eta$ orthogonal to $X$. 

Moreover, if $\Y$ is  timelike, the last two terms are zero if and only if $g$-dual one-form of $X$ is closed and $\nabla \eta \cdot X=0$. 
\end{enumerate}
\end{proposition}
\begin{proof}
Using the  coordinates from \eqref{eq:g} and $\Y = \tfrac{\partial}{\partial u}$, we have  
$\Y^\alpha_{|\beta}=  \Y^\alpha_{,\beta} + \bm \Gamma_{\beta \gamma}^\alpha \Y^\gamma = \bm\Gamma_{\beta u}^\alpha$  and 
\begin{align*}
	\Y^u_{|u} &= \bm \Gamma_{uu}^u =  \tfrac{1}{2f} \nabla \eta \cdot \hat X\\
	\Y^u_{|i} &=\bm  \Gamma_{iu}^u = \tfrac{1}{2f^2} \eta_{,i} + \tfrac{1}{2f } \hat X^k(X_{k;i} - X_{i;k})\\
	\Y^i_{|u} &=\bm  \Gamma_{uu}^i = \tfrac{1}{2} \g^{ij} \eta_{,j} \\
	\Y^i_{|j} &=\bm \Gamma_{ju}^i = \tfrac{1}{2}\g^{ik} (X_{k;j} - X_{j;k}) - \tfrac{1}{2f} \hat X^i \eta_j.
\end{align*}
From \eqref{eq:grad}, we prove \eqref{eq:nablaY}:
\begin{align}\label{eq:gradY}
\begin{split}
	&\g (\bm \nabla \Y, \bm \nabla \Y)  = -(\Y^u_{;u})^2 - 2 \Y^u_{;i} \Y^i_{;u} - \Y^i_{;j} \Y^j_{;i}\\
	&\; = - \tfrac{1}{4f^2}(\nabla \eta\cdot \hat X)^2 - 2\Big(\tfrac{1}{2f^2} \eta_{,i} + \tfrac{1}{2f } \hat X^k (X_{k;i} - X_{i;k}) \Big)\Big( \tfrac{1}{2}\g^{ij} \eta_j\Big)\\
	& \quad \quad - \Big(  \tfrac{1}{2}\g^{ik} (X_{k;j} - X_{j;k}) - \tfrac{1}{2f} \hat X^i \eta_j\Big) \Big(\tfrac{1}{2} \g^{jk} (X_{k;i} - X_{i;k}) - \tfrac{1}{2f}\hat X^j \eta_i\Big)\\
&\; = - \tfrac{1}{2f^2} |\nabla \eta|^2 -  \tfrac{1}{f} (\nabla \eta)^i\hat X^k (X_{k;i} - X_{i;k})  +  \tfrac{1}{4} \g^{ik} (X_{k;j} - X_{j;k}) \g^{j\ell} (X_{ i; \ell} - X_{\ell;i}).
\end{split}
\end{align}

We compute $\g^{ik} (X_{k;j} - X_{j;k}) \g^{j\ell} (X_{ i; \ell} - X_{\ell;i})$ in a local orthonormal frame ${e_1, \dots, e_n}$ on $U$ with $e_n = \frac{X}{|X|_g}$ whenever $X\neq 0$:
\begin{align*}
    &\g^{ik} (X_{k;j} - X_{j;k}) \g^{j\ell} (X_{ i; \ell} - X_{\ell;i}) \notag \\
    & =( \delta^{ik} -|\hat X|^2\delta^{in} \delta^{kn}) (X_{k;j} - X_{j;k}) (\delta^{j\ell} - |\hat X|^2 \delta^{jn} \delta^{\ell n} )(X_{ i; \ell} - X_{\ell;i}) \notag\\
    &=\sum_{i,j=1}^n(X_{i;j}- X_{j;i})^2 - 2 |\hat X|^2\sum_{j=1}^n (X_{n;j} - X_{j;n})^2\notag\\
    &= \sum_{a,b=1}^{n-1}(X_{a;b}- X_{b;a})^2 +  2  (1-|\hat X|^2) \sum_{a=1}^{n-1}(X_{n;a} - X_{a;n})^2.
\end{align*}
Therefore, if $\Y$ is null, the vanishing of the above term implies $X_{a;b}=X_{b;a}$ for all $a, b=1, \dots, n-1$, and hence $X^\perp$ is involutive. If $\Y$ is timelike; that is, $|\hat X| <1$, then the vanishing of the above term implies $X_{i;j}=X_{j;i}$ for all $i,j=1,\dots, n$. 

Next, assuming  $\eta \neq 0$, the last two terms in the right hand side of \eqref{eq:gradY} become:
\begin{align*}
&- \tfrac{1}{f} (\nabla \eta)^i\hat X^k (X_{k;i} - X_{i;k}) +  \tfrac{1}{4} \g^{ik} (X_{k;j} - X_{j;k}) \g^{j\ell} (X_{ i; \ell} - X_{\ell;i})\\
&= -\sum_{a=1}^{n-1}\tfrac{1}{f^2} (\nabla \eta^\perp )^a |X|_g (X_{n;a} - X_{a;n}) \\
&\quad +  \tfrac{1}{4}\sum_{a,b=1}^{n-1}(X_{a;b}- X_{b;a})^2 +  \tfrac{1}{2}  (1-|\hat X|^2) \sum_{a=1}^{n-1}(X_{n;a} - X_{a;n})^2\\
& = - \tfrac{|\hat X|^2 }{2f^2 (1-|\hat X|^2)}  |\nabla \eta^\perp|^2 \\
&\quad +\tfrac{1}{2} \Big( - \tfrac{|X|}{f^2 (1-|\hat X|^2)^{\frac{1}{2}} } \nabla \eta^\perp + (1-|\hat X|^2)^{\frac{1}{2}} \sum_{a=1}^{n-1} (X_{n;a} - X_{a;n})e_a \Big)^2\\
&\quad +  \tfrac{1}{4}\sum_{a,b=1}^{n-1}(X_{a;b}- X_{b;a})^2.
\end{align*}
By substituting $|\nabla \eta^\perp|^2 = |\nabla \eta|^2 - |\hat X|^{-2} (\nabla \eta \cdot \hat X)^2$, we can rewrite the first term on the right hand side of the previous identity as  
\[
- \tfrac{|\hat X|^2 }{2f^2 (1-|\hat X|^2)}  |\nabla \eta^\perp|^2 = - \tfrac{1}{2\eta} (|\nabla \eta|^2 |\hat X|^2 - (\nabla \eta \cdot \hat X)^2).
\]
Combining this with the first term on the right-hand side of \eqref{eq:gradY}  yields
\begin{align*}
	- \tfrac{1}{2f^2} |\nabla \eta|^2 - \tfrac{1}{2\eta} (|\nabla \eta|^2 |\hat X|^2 - (\nabla \eta \cdot \hat X)^2) = -\tfrac{1}{2\eta} (|\nabla\eta|^2 - (\nabla \eta\cdot \hat X)^2).
\end{align*}
Then the desired identity follows. 
\end{proof}

\begin{remark}
We have expressed the first two terms in \eqref{eq:bochner} using initial data. The  term $\Ric_{\g}(\Y, \Y)$  there can also be expressed using initial data as
\begin{align*}
	\Ric_{\g}(\Y, \Y) &= f(\Delta_g f + |k|_g^2 f + \tr_g (\mathscr L_X k) +2 g(X, J)) \\
	&\quad  + X^i X^j\left(R_{ij} + (\tr_g k) k_{ij} - 2k_{i\ell} k^{\ell}_j - f^{-1} (f_{;ij} +(\mathscr L_X k)_{ij} )\right) 
\end{align*}
by the computations in \cite[Lemma B.1]{Huang-Lee:2024}.
\end{remark}

Recall that \eqref{eq:bochner} can be written in terms of $\eta = -\g(\Y,\Y)$ as
\begin{align} \label{eq:box-eta}
	\Box_{\g} \eta = -2\g(\bnabla \Y, \bnabla \Y) + 2 \Ric_{\g} (\Y, \Y).
\end{align}
As direct consequences of Lemma~\ref{le:elliptic1} and Proposition~\ref{pr:elliptic2}, we obtain identities for $\eta$ using the initial data, which play a fundamental role in this paper.

\begin{corollary}
Let  $(\N, \g)$ be a spacetime admitting a  Killing vector field $\Y$ and an initial data set $(U, g, k)$. Let $(f, X)$ be the lapse-shift pair of $\Y$ and suppose $f\neq 0$. Then 
\begin{align}\label{eq:equ}
\begin{split}
\Div_g \big(  f (\nabla \eta -( \nabla \eta \cdot \hat X)\hat X)\big)&=  \tfrac{1}{f} |\nabla \eta|^2+2  (\nabla \eta)^i\hat X^k (X_{k;i} - X_{i;k}) - \phi_1 
\end{split}
\end{align}
where 
\[
 \phi_1 =   \tfrac{1}{2} f\g^{ik} (X_{k;j} - X_{j;k}) \g^{j\ell} (X_{ i; \ell} - X_{\ell;i}) - 2f \Ric_{\g} (\Y, \Y).
 \] 
 If, furthermore,  $\eta\neq 0$ and we set the vector field $V =  f (\nabla \eta -( \nabla \eta \cdot \hat X)\hat X) $, then 
 \begin{align} \label{eq:div-form}
 	\Div_g V=  \tfrac{ 1}{\eta}V\cdot \nabla \eta  - \phi_2
 \end{align}
where 
\begin{align*}
	\phi_2 &=  f\Big(- \tfrac{|X|}{f^2 (1-|\hat X|^2)^{\frac{1}{2}} } \nabla \eta^\perp +(1-|\hat X|^2)^{\frac{1}{2}}\sum_{a=1}^{n-1} (X_{n;a} - X_{a;n})e_a \Big)^2\\
	&\quad  +  \tfrac{1}{2}f\sum_{a,b=1}^{n-1}(X_{a;b}- X_{b;a})^2 - 2f \Ric_{\g} (\Y, \Y).
\end{align*}
\end{corollary}

%\begin{proof}

%Using \eqref{eq:bochner}, \eqref{eq:elliptic}, and \eqref{eq:nablaY}, we obtain
%\begin{align*}
%\Div_g \big(  f (\nabla \eta -( \nabla \eta \cdot \hat X)\hat X)\big)&\le   \tfrac{1}{f} |\nabla \eta|^2  +   2 (\nabla \eta)^i\hat X^k (X_{k;i} - X_{i;k}) \\
%&=  \Big( \tfrac{1}{f} (\nabla \eta)_i  +   2 \hat X^k (X_{k;i} - X_{i;k})\Big) \big(\nabla \eta\big)^i
%\end{align*}

%Using  \eqref{eq:bochner}  and \eqref{eq:nablaY2}, we obtain
%\begin{align*}
%	\Div_g \big(  f (\nabla \eta -( \nabla \eta \cdot \hat X)\hat X )\big)&\le \tfrac{1}{f} |\nabla \eta|^2  + \tfrac{f}{\eta} \big(|\nabla \eta|^2 |\hat X|^2 - (\nabla \eta \cdot \hat X)^2\big)\\
%	&= \tfrac{ f}{\eta} \big(\nabla \eta - (\nabla \eta \cdot \hat X) \hat X\big)\cdot \nabla \eta.
%\end{align*}
%\end{proof}

%Observe that the equation is elliptic if $\Y$ is causal, i.e., equivalently, if $\eta \ge 0$. The equation becomes degenerately elliptic precisely on the null set  $Z=\{ p\in \N: \eta(p)=0\} = \{ p \in \N: |\hat X(p)|=1\}$. (Recall $\hat X = \frac{X}{f}$.) 

%As mentioned above, this proposition can also be expressed entirely in terms of initial data sets, cf. Appendix \ref{Appendix D}.

We prove Theorem~\ref{TH:monotonicity}, which follows from the more general theorem below.

\begin{theorem}[Monotonicity formula]\label{th:monotonicity}
Let $(\N, \g)$ be a spacetime admitting a causal Killing vector field $\Y$ with $\Ric_{\g}(\Y, \Y)\le 0$. Let $(U, g, k)$ be an initial data set in $\N$. Let $(f, X)$ be the lapse–shift pair of $\Y$ with $f>0$.  

For regular values $t>0$ of $\eta$, assume the level sets  $\Sigma_t = \{ x \in U : \eta(x) = t \}$ are compact without boundary. Define 
\begin{align}
 F(t) &= \int_{\Sigma_t} f\big(\nabla \eta - (\nabla \eta \cdot \hat X)\hat X\big)\cdot \tfrac{\nabla \eta}{|\nabla \eta|}\, d\sigma, \\
 M(t) &= t^{-1}F(t) = t^{-1}  \int_{\Sigma_t} f\big(\nabla \eta - (\nabla \eta \cdot \hat X)\hat X\big)\cdot \tfrac{\nabla \eta}{|\nabla \eta|}\, d\sigma.\notag
 \end{align}
 Then
\begin{enumerate}
\item $M(t)\ge \int_{\Sigma_t} f^{-1} |\nabla \eta| \, d\sigma$ for  regular values $t$ of $\eta$, and $M(t) \in W^{1,1}_{\loc} ((0, \infty))$.
\item  $M(t)$ is non-increasing in  $t$. \label{it:mono}
\item  \label{it:con} $M(t)$ is constant for $t\in (a,b)$ if and only if,  on the region $\Omega_b\setminus \Omega_a$, we have  $\Ric_{\g}(\Y, \Y)=0$, the $g$-dual one-form of $X$ is closed, and $\nabla \eta \cdot X=0$.
\item \label{it:limit} If $\lim_{t\to 0^+} F(t) = 0$, then  
\[
 \lim_{t\to 0^+} M(t) \le \limsup_{t\to 0^+} \int_{\Sigma_t} \left(\tfrac{1}{f}(\nabla \eta)_i - 2\hat X^k (X_{k;i} - X_{i;k})\right)\tfrac{(\nabla \eta)^i}{|\nabla \eta|}\, d\sigma.
\]
\end{enumerate}
\end{theorem}
\begin{proof}

Applying the Cauchy-Schwarz inequality yields, for a regular value $t$ of $\eta$, 
\begin{align*}
	M(t) &= \int_{\Sigma_t}\tfrac{1}{\eta |\nabla \eta|} f(|\nabla \eta|^2 - (\nabla \eta \cdot \hat X)^2)\, d\sigma\\
	& \ge \int_{\Sigma_t} \tfrac{1}{\eta |\nabla \eta|} f(1-|\hat X|^2) |\nabla \eta|^2 \, d\sigma= \int_{\Sigma_t} f^{-1} |\nabla \eta| \, d\sigma.
\end{align*}

In the following computation, we denote the vector field 
\[
	V =  f (\nabla \eta -( \nabla \eta \cdot \hat X)\hat X).
\]
Define the sublevel set $\Omega_t = \{ x\in U: \eta(x) < t\}$. For any regular values $0< a<t<b$ of $\eta$, we apply the divergence theorem in the first line and the coarea formula in the second line below to obtain
\begin{align*}
 F(t) &= F(a) +  \int_{\Omega_{t}\setminus \Omega_{a} }  \Div_g V \, d\mu\\
 &= F(a) + \int_a^t \int_{\Sigma_s}  \tfrac{1}{|\nabla \eta|} \Div_g V \, d\sigma ds.
 \end{align*}
 It implies that $F\in L^1([a,b])$ and for almost every $t$, 
 \[
 F'(t) = \int_{\Sigma_t}  \tfrac{1}{|\nabla \eta|} \Div_gV\, d\sigma \in L^1([a,b]). 
\]
It is standard to verify that $F'(t)$ is the weak derivative of $F$, which shows that $F(t)\in W^{1,1}([a,b])$, and consequently $M(t) \in W^{1,1}([a,b])$ as well. 

%To see that $B'(t)\in W^{1,1}([a,b])$, let $\chi\in \C^\infty_c((a,b))$ be a test function, we compute 
%\begin{align*}
%	 - \int_a^b \chi'(t) B(t) \, dt &= -\int_a^b  \int_{\Sigma_t} \chi'(\eta)V\cdot \tfrac{\nabla \eta}{|\nabla \eta| } \, d\sigma\\
%	 &= - \int_{\Omega_b \setminus \Omega_a } \chi'(\eta) V\cdot \nabla \eta \, dv\qquad  (\mbox{by coarea formula})\\
%	 &= - \int_{\Omega_b \setminus \Omega_a }  V\cdot \nabla (\chi(\eta))\, dv \quad \\
%	 &= \int_{\Omega_b \setminus \Omega_a }  \Div V (\chi(\eta))\, dv\\
%	 &= \int_a^b \chi(t) \int_{\Sigma_t}  \tfrac{1}{|\nabla \eta|} \Div V \,d\sigma dt \qquad (\mbox{by coarea formula}).
%\end{align*}

To prove Item~\eqref{it:mono}, we apply \eqref{eq:div-form} and the divergence theorem gives
\begin{align}\label{eq:B}
\begin{split}
 F(t) - F(a)&= \int_{\Omega_{t}\setminus\Omega_a }  \Div_g V \, d\mu\\
 &= \int_{\Omega_{t}\setminus \Omega_a } \tfrac{ 1}{\eta}V \cdot \nabla \eta\, d\mu - \int_{\Omega_{t}\setminus \Omega_a } \phi_2 \, d\mu
 \\
	&= \int_{a}^{t} \int_{\Sigma_s} \tfrac{ 1}{\eta}V\cdot \tfrac{\nabla \eta}{|\nabla \eta|}  \, d\sigma ds- \int_{\Omega_{t}\setminus \Omega_a } \phi_2 \, d\mu\\
	&= \int_{a}^{t} M(s) \, ds  -\int_{\Omega_{t}\setminus \Omega_a } \phi_2 \, d\mu.
\end{split}
\end{align}
Hence,
\begin{align} \label{eq:M}
	tM(t)  -a M(a) = \int_{a}^{t} M(s)\, ds-\int_{\Omega_{t}\setminus \Omega_a } \phi_2 \, d\mu.
\end{align}
It follows that for almost every $t>0$,  
\[
M(t) + tM'(t) \le M(t) \quad \Longrightarrow \quad M'(t)\le 0.
\] 
 Because $M(t)\in W^{1,1}_{\loc}$, it has an absolutely continuous representative, which implies that $M(t)$ is non-increasing.

For Item~\eqref{it:con}, we analyze \eqref{eq:M}. If $M(t)$ is constant on $[a,b]$, the $\phi_2 = 0$ on $\Omega_{a}\setminus \Omega_b$, and the desired characterization follows from Proposition~\ref{pr:elliptic2}. Conversely, if $\phi_2=0$ on $\Omega_{a}\setminus \Omega_b$, then $M'(t)\equiv 0$, and thus $M(t)$ is constant.

For Item~\eqref{it:limit}, we integrate \eqref{eq:equ} on $\Omega_t \setminus \Omega_{a}  $ and apply the divergence theorem and co-area formula,
\begin{align*}
	F(t) - F(a) &=  \int_{\Omega_t\setminus \Omega_{a} }  \Div_g \big(  f (\nabla \eta -( \nabla \eta \cdot \hat X)\hat X)\big) \, d\mu\\
	&\le \int_{\Omega_t\setminus \Omega_{a} } \left(  \tfrac{1}{f} (\nabla \eta)_i +   2 \hat X^k (X_{k;i} - X_{i;k})\right)  (\nabla \eta)^i \, d\mu\\
	&=  \int_{a}^t \int_{\Sigma_s} \left(  \tfrac{1}{f} (\nabla \eta)_i +   2 \hat X^k (X_{k;i} - X_{i;k})\right) \tfrac{(\nabla \eta)^i}{|\nabla \eta|}\, d\sigma ds.
\end{align*}
If $\lim_{a\to 0^+} F(a)=0$, then 
\begin{align*}
	M(t) &= t^{-1}F(t)\le   t^{-1} \int_{0}^t \int_{\Sigma_s} \left(  \tfrac{1}{f} (\nabla \eta)_i +   2 \hat X^k (X_{k;i} - X_{i;k})\right) \tfrac{(\nabla \eta)^i}{|\nabla \eta|}\, d\sigma ds.
\end{align*}
The desired result follows L'H\^opital's rule. 
\end{proof}

\begin{remark}
In the case $X\equiv 0$, that is, when $\Y$ is orthogonal to $U$,  we have $M(t) =2\int_{\Sigma_t} |\nabla f|\, d\sigma $, and the previous theorem recovers the classical fact that if $\Delta f\le 0$, then $\int_{\Sigma_t} |\nabla f|$ is non-increasing. 
\end{remark}

\begin{proposition}[Localized monotonicity formula]\label{pr:mono}
Let  $(\N, \g)$ be a spacetime admitting a causal Killing vector field $\Y$ with $\Ric_{\g}(\Y, \Y)\le 0$.  Let $(U, g, k)$ be an initial data set in $\N$, and let $(f, X)$ be the lapse-shift pair of $\Y$ with $f>0$. Let $K \subset U$ be a compact smooth $n$-dimensional manifold with 
piecewise $\C^1$ boundary such that $\partial K = \partial_1 K \sqcup \partial_2 K$, where $\partial_1 K$ consists of regular level sets of $\eta$ and $\partial_2 K$ consists of hypersurfaces  tangent to $\nabla \eta - (\nabla \eta \cdot \hat{X}) \hat{X}$.  Suppose $K$ is foliated by level sets of $\Sigma_t$ transvese to $\partial_2 K$.

Denote the level set of $\eta$ on $U$ by $\Sigma_t = \{ x\in  U: \eta (x) = t\}$. For regular values $t>0$ of $\eta$, we define 
\begin{align*}
 F_K(t) &= \int_{\Sigma_t\cap K }f (\nabla \eta - (\nabla \eta \cdot \hat X)\hat X) \cdot \tfrac{\nabla \eta}{|\nabla \eta|}\, d\sigma \\
 M_K(t) &= t^{-1}F_K(t).
\end{align*}
Then 
\begin{enumerate}
\item  $M_K(t) \ge \int_{\Sigma_t} f^{-1} |\nabla \eta| \, d\sigma$ for regular values $t$ of $\eta$, and $M_K(t)\in W^{1,1}_{\loc}((0, \infty))$. 
\item $M_K(t)$ is non-increasing in $t$. 
\item $M_K(t)$ is constant for $t\in (a,b)$ if and only if, on the region $(\Omega_b\setminus \Omega_a )\cap K$, we have $\Ric_{\g}(\Y, \Y)=0$, the $g$-dual one-form of $X$ is closed, and $\nabla \eta \cdot X=0$. 
\item If $\lim_{t\to 0^+} F_K(t) =  0$, then 
\begin{align*}
	 \lim_{t\to 0^+} M_K(t) \le  \limsup_{t\to 0^+} \int_{\Sigma_t\cap K}  \left( \tfrac{1}{f} (\nabla \eta)_i  -   2 \hat X^k (X_{k;i} - X_{i;k})\right)  \tfrac{(\nabla \eta)^i}{|\nabla \eta|} \, d\sigma.
\end{align*}
\end{enumerate}
\end{proposition}
\begin{proof}
The proof follows a similar argument as in the proof of Theorem~\ref{th:monotonicity}, except that we now integrate  $ \Div_g \big(  f (\nabla \eta -( \nabla \eta \cdot \hat X)\hat X)\big) $ over the set $(\Omega_t \setminus \Omega_a) \cap K$.  Applying the divergence theorem to $(\Omega_t \setminus \Omega_a) \cap K$ introduces an additional boundary integral term on $\partial_2 K$, which vanishes because we assume $\partial_2 K$ is tangent to $\nabla \eta - (\nabla \eta \cdot \hat{X}) \hat{X}$. For example, in the computations parallel to \eqref{eq:B}, we have, for $\Sigma_t, \Sigma_a $ both intersecting $K$, 
\begin{align*}
F_K(t) - F_K(a)&= \int_{(\Omega_t\setminus\Omega_a ) \cap K }  \Div_g \big(  f (\nabla \eta -( \nabla \eta \cdot \hat X)\hat X)\big) \, d\sigma\\
&\quad  + \int_{(\Omega_t\setminus\Omega_a ) \cap \partial_2 K} f (\nabla \eta - (\nabla \eta \cdot \hat X)\hat X) \cdot \nu\, d\sigma \\
&= \int_{(\Omega_t\setminus\Omega_a ) \cap K }  \Div_g \big(  f (\nabla \eta -( \nabla \eta \cdot \hat X)\hat X)\big) \, d\sigma
\end{align*}
where $\nu$ is the outward unit normal to $\partial_2 K$. The rest of the argument is similar. 
\end{proof}

\section{Uniform Stationarity}\label{se:st}

We let $(M, g, k)$ be an asymptotically flat initial data set in a spacetime $(\N, \g)$ admitting a Killing vector field $\Y$. We first focus on the case 
 where $\Y$ is asymptotically null, so the lapse-shift pair $(f, X)$ on $M$ is asymptotic to a null vector $(a,b)$. The first lemma shows that $\eta := -\g(\Y, \Y)$ exhibits a better  decay rate, using the equation
 \begin{align}\tag{\ref{eq:equ}}
\begin{split}
&\Div_g \big(  f (\nabla \eta -( \nabla \eta \cdot \hat X)\hat X)\big)=  \tfrac{1}{f} |\nabla \eta|^2+2  (\nabla \eta)^i\hat X^k (X_{k;i} - X_{i;k}) \\
&\qquad \qquad -\tfrac{1}{2} f\g^{ik} (X_{k;j} - X_{j;k}) \g^{j\ell} (X_{ i; \ell} - X_{\ell;i}) + 2f \Ric_{\g} (\Y, \Y),
\end{split}
\end{align}
where $\hat X = \frac{X}{f}$.

 \begin{lemma}
 Let  $(\N, \g)$ be a spacetime admitting a Killing vector field $\Y$.  Let $(M, g, k)$ be an asymptotically flat initial data set in $(\N, \g)$, and let $(f, X)$ be the lapse-shift pair of $\Y$ that is asymptotically translational to $(a, b)$ where $a= |b|\neq 0$. If the Einstein tensor $G= O(|x|^{-n-\epsilon})$ for some $\epsilon>0$, then for some $q_1>0$,  
 \begin{align}\label{eq:eta}
 	\eta(x) = O^{2,\alpha}(|x|^{2-n-q_1}). 
 \end{align}
 \end{lemma}
 \begin{proof}
 In the proof, we let $q_1>0$, which may vary from equation to equation.  From the expansion of $(f, X)$ in Theorem~\ref{th:expan} and Remark~\ref{re:Rij}, we obtain, for some constant $A$, 
 \[
 \eta(x) = A |x|^{2-n} +  O^{2,\alpha}(|x|^{2-n-q_1}) 
 \]    We claim that the right hand side of \eqref{eq:equ} is of the order $O^{0,\alpha}(|x|^{-n-q_1})$.  Since the quadratic terms decay faster, it suffices to check the decay of $\Ric_{\g}(\Y, \Y)$.  The assumption  $G= O^{0,\alpha}(|x|^{-n-q_1})$ implies  $\Ric_{\g} =O^{0,\alpha}(|x|^{-n-q_1})$, and therefore, 
 \[
 \Ric_{\g} (\Y, \Y) = f^2 \Ric_{\g} (\n, \n) +2 f\Ric_{\g} (\n, X) + \Ric_{\g} (X, X) = O^{0,\alpha}(|x|^{-n-q_1}).
 \]
 On the other hand, we assume $\hat X = \partial_n + O^{2,\alpha}(|x|^{-q})$ by rotating the coordinates and compute the left hand side of \eqref{eq:equ} as
 \begin{align*}
 &\Div_g \big(  f (\nabla \eta -( \nabla \eta \cdot \hat X)\hat X)\big) \\
 &= a (\Delta_0 \eta - \partial_n^2 \eta) + O^{0,\alpha}(|x|^{-n-q_1})\\
 &= -a A (2-n) |x|^{-n} (1- n|x|^{-2} (x^n)^2) + O^{0,\alpha}(|x|^{-n-q_1}).
 \end{align*}
Comparing both sides shows that $A=0$. \end{proof}

\begin{proposition}\label{pr:asymp}
Let  $(\N, \g)$ be a spacetime admitting a Killing vector field $\Y$.  Let $(M, g, k)$ be an asymptotically flat initial data set in $(\N, \g)$, and let $(f, X)$ be asymptotically translational to $(a, b)$ where $a=|b|\neq 0$. Let  $\{\Omega_k\}_{k\in \mathbb N}$  be an exhaustion of $M$ by compact manifolds with (possibly piecewise smooth) boundary. Then 
\begin{align}
	\lim_{k \to \infty} \int_{\partial \Omega_k\setminus \partial M } f (\nabla \eta - (\nabla \eta \cdot \hat X) \hat X) \cdot \nu \, d\sigma&=0 \label{eq:bdry}\\
	\lim_{k \to \infty} \int_{\partial \Omega_k\setminus \partial M } \left(\tfrac{1}{f} (\nabla \eta)_i + 2 \hat X^k (X_{k;i} - X_{i;k})  \right)\nu^i \, d\sigma &= 0 \label{eq:B'}
\end{align}
where $\nu$ is the outward unit normal vector to $\partial \Omega$. 
\end{proposition}
\begin{proof}
For a flux integral, the divergence theorem implies that its limit coincides with that of the interior integral and is independent of the chosen compact exhaustion. Therefore, we may compute the limits using a convenient compact exhaustion.

To prove \eqref{eq:bdry}, we use  large coordinate balls $\{ B_r \}$ as a compact exhaustion. By \eqref{eq:eta}, the integrand $f (\nabla \eta - (\nabla \eta \cdot \hat X) \hat X) \cdot \nu$ is  $O(r^{1-n-q_1})$ on $\partial B_r$, while  $\mathrm{vol}(\partial B_r) = O(r^{n-1})$, 
\[
\lim_{r \to \infty} \int_{\partial B_r } f (\nabla \eta - (\nabla \eta \cdot \hat X) \hat X) \cdot \nu \, d\sigma = \lim_{r\to \infty} O(r^{-q_1}) = 0.
\]	

To prove \eqref{eq:B'}, we use large coordinate cylinders $C_\rho = \{ x\in M: |x'|\le \rho, |x^n|\le \rho^s\}$, with the number $s$ fixed such that $\frac{n-1}{n-1+q_1} <s<1$ and  $x' = (x^1, \dots, x^{n-1})$.  The first integrand $\frac{1}{f} \nabla \eta \cdot \nu$ on $\partial C_\rho$ is of the order $O(r^{1-n-q_1})=O\big(\rho^{s(1-n-q_1)}\big)$, while of $\mathrm{vol}(\partial C_\rho) =O(\rho^{n-1})$, so
\[
	\lim_{\rho \to \infty} \int_{\partial C_\rho } \tfrac{1}{f} \nabla \eta \cdot \nu \, d\sigma = \lim_{\rho \to \infty} O(\rho^{s(1-n-q_1) + n-1}) = 0.
\]

To handle the second integrand $ \hat X^k (X_{k;i} - X_{i;k}) \cdot \nu $ in \eqref{eq:B'},  we rotate coordinates so that $b= a\partial_n$, and  decompose $\partial C_\rho =\partial_1 C_\rho \cup \partial_2 C_\rho$ where  $\partial_1 C_\rho$ consists the bases of the cylinder and $\partial_2 C_\rho$ is the lateral hypersurface: 
\begin{align*}
\partial_1 C_\rho &= \big\{ x\in M: |x^n|=\rho^s, |x'|\le \rho\big\}   \\
 \partial_2 C_\rho &= \big\{ x\in M: |x^n|\le \rho^s, |x'|= \rho\big\}. 
\end{align*}
These have volumes: 
\begin{align*}
\mathrm{vol}(\partial_1 C_\rho) &=O(\rho^{n-1})\quad \mbox{ and } \quad \mathrm{vol}(\partial_2 C_\rho) = O(\rho^{n-2+s}).
\end{align*}
On $\partial_1 C_\rho$, 
\[ 
\hat X^k (X_{k;i} - X_{i;k}) \nu^i = -(\partial_n X^{n} -\partial_n X^{n}) + O(r^{1-n-q_1}) = O\big(\rho^{s(1-n-q_1)}\big),
\]
so
\begin{align*}
\lim_{\rho \to \infty} \int_{\partial_1 C_\rho}\hat X^k (X_{k;i} - X_{i;k})  \nu^i \, d\sigma &=\lim_{r\to \infty} O(r^{s(1-n-q_1)+ (n-1)}) = 0. 
\end{align*}
On  $\partial_2 C_\rho$, we have $\rho \le r \le  2\rho $, and  Theorem~\ref{th:expan} gives 
\begin{align*}
	\hat X^k (X_{k;i} - X_{i;k}) \nu^i &= - \sum_{a=1}^{n-1}(\partial_a X^{n} -\partial_n X^{a}) \tfrac{x^a}{\rho}+ O(\rho^{1-n-q_1})= O(\rho^{1-n}),
\end{align*}
so
\begin{align*}
\lim_{\rho \to \infty} \int_{\partial_2 C_\rho}  \hat X^k (X_{k;i} - X_{i;k}) \nu^i \, d\sigma &= \lim_{\rho \to \infty} O(\rho^{(n-2+s)+(1-n)}) =0.
\end{align*}
\end{proof}

We combine the above results to prove Theorem~\ref{Th:uni}.
\begin{proof}[Proof of Theorem~\ref{Th:uni}]
By Lemma~\ref{le:asy} and that $\Y$ is timelike, $(f, X)$ must be asymptotically translational to $(a, b)$. If instead $\Y$ were asymptotically null,  then $a=|b|\neq0$. As in Theorem~\ref{th:monotonicity}, let $\Sigma_t $ denote the level sets of $\eta$ in $M$ and define
\begin{align*}
	F(t) &= \int_{\Sigma_t} f (\nabla \eta - (\nabla \eta \cdot \hat X) \hat X) \cdot \tfrac{\nabla \eta}{|\nabla \eta|} \, d\sigma\\
	M(t) &= t^{-1} F(t).
\end{align*}
For sufficiently small $t>0$, $\Sigma_t$ is a compact hypersurface without boundary. By \eqref{eq:bdry}, $\lim_{t\to0^+}F(t)=0$. Hence, by Item~\eqref{it:limit} in Theorem~\ref{th:monotonicity} and \eqref{eq:B'},  $\lim_{t\to0^+}M(t)=0$. Since  $M(t)$ is non-increasing in $t$, it  follows that  $\nabla \eta=0$, and thus $\eta\equiv0$.  A contradiction. \end{proof}

\section{The Strong Maximum Principle}\label{se:max}

Let $(\N, \g)$ be a spacetime having a Killing vector field $\Y$. Denote $\eta = - \g(\Y, \Y  )$ and the (spacetime) null set $\mathbf Z= \{ p\in  \N: \eta(p)=0\}$. The null set is intimately connected with  the concept of Killing (pre)horizons. However, we do not impose additional assumptions,  such as connectedness, foliation by null hypersurfaces, absence of zeros, as are typically required for Killing (pre)horizons.

Let $(f, X)$ be the lapse-shift pair of $\Y$ along an initial data set $(U, g, k)$ with $f>0$.  We recall that $\eta$ satisfies  the differential equation on $U$:
\begin{align}\tag{\ref{eq:equ}}
\begin{split}
\Div_g \big(  f (\nabla \eta -( \nabla \eta \cdot \hat X)\hat X)\big)&=  \tfrac{1}{f} |\nabla \eta|^2+2  (\nabla \eta)^i\hat X^k (X_{k;i} - X_{i;k}) - \phi_1 
\end{split}
\end{align}
where $\hat X = \frac{X}{f}$ and $\phi_1 =   \tfrac{1}{2} f\g^{ik} (X_{k;j} - X_{j;k}) \g^{j\ell} (X_{ i; \ell} - X_{\ell;i}) - 2f \Ric_{\g} (\Y, \Y).$

Define the associated \emph{linear} differential operator
\begin{align*}
	Lw &:= \Div_g \big( f (\nabla w - (\nabla w \cdot \hat X)\hat X) \big) \\
	&\quad - \tfrac{1}{f} \nabla \eta\cdot \nabla w-2 (\nabla w)^i \hat X^k (X_{k;i} - X_{i;k}). 
\end{align*}
The principal coefficients of  $L$ are given by $\g^{ij}  = g^{ij} - \hat X^i \hat X^j$. We rewrite  \eqref{eq:equ} using $L$ and get
\begin{align} \label{eq:L}
 L\eta=- \tfrac{1}{2} f \g^{ik} (X_{k;j} - X_{j;k}) \g^{j\ell} (X_{i;\ell} - X_{\ell;i}) + 2f\Ric_{\g}(\Y, \Y). 
 \end{align}

If $\Y$ is a causal Killing vector field,  then $\eta = f^2 - |X|^2 \ge 0$ and $|\hat X|\le 1$.  Define the \emph{null set}  $Z = \{ x\in U: \eta(x) = 0\} = \mathbf Z\cap  U$. The operator $L$ is  strictly elliptic on $U\setminus Z$, while it degenerates on $Z$ with characteristic direction $\hat X$; that is, $\g^{ij} \hat X_i \hat X_j=0$ on $Z$.  

To establish a strong maximum principle for $\eta$, we first analyze the null set $Z$, with the aim of proving the regularity of $\partial Z$,  Lemma~\ref{Le:smooth}, through the next three results.

We work  in a local double null frame in the next lemma. Given a null vector $\ell$ in a spacetime $(\N, \g)$, its conjugate null vector $\underline \ell$ is defined by $\g(\ell, \underline \ell)=1$. Consider a local double null frame $\{ E_0 , E_1, E_2, \dots, E_{n}\}$ of $(\N, \g)$, where $E_0=\ell$ is a null vector, $E_1=\underline \ell$, and $\{E_2, \dots, E_{n}\}$ is a local orthonormal set orthogonal to both $\ell$ and $\underline \ell$. In this frame, the only nonzero components of $\g$ and $\g^{-1}$  are  $\g_{01} = \g_{10} = \g_{aa} = 1$, and $\g^{01} = \g^{10} = \g^{aa} = 1$ for each $a = 2, \dots, n$.

\begin{lemma}\label{le:Y}
Let $(\N, \g)$ be a spacetime admitting a Killing vector field~$\Y$. Suppose that $\eta(p)=0$ and  $\bnabla \eta (p)=0$ for some point $p\in \N$. Then,  at $p$, for all $a, b = 2, \dots, n$,  
\begin{align}\label{eq:Yall}
\left\{ 
\begin{array}{ll}
\bnabla_{E_a} \Y &=  \g(\bnabla_{E_a} \Y, \underline \ell) \Y + \g(\bnabla_{E_a} \Y, E_b) E_b\\
 \bnabla_\Y \Y &=0  \\
 \bnabla_{\underline \ell } \Y &  = -  \g(\Y, \bnabla_{\underline \ell} E_a  ) E_a %= \g(\bnabla_{\bar \Y} \Y, E_a ) E_a
\end{array}\right.
 \end{align}
where $\{ \Y, \underline\ell, E_2, \dots, E_n\}$ is a local double null frame.
 
 If in addition we assume  $\Box_{\g} \eta\ge 0$ and $\Ric_{\g}(\Y, \Y)\le 0$  at $p$, then, at $p$, we have $\Box_{\g} \eta=\Ric_{\g}(\Y,\Y)=0$ and 
\begin{align}\label{eq:Y2}
	\Y_{b| a} := \g(\bnabla_{E_a} \Y, E_b)=0 \quad \mbox{  for } a, b = 2, \dots, n.
\end{align}
Furthermore, $\Y$ does not vanish at $p$. 
\end{lemma}
\begin{proof}
It is straightforward to verify \eqref{eq:Yall}. To prove \eqref{eq:Y2}, we use ~\eqref{eq:box-eta} to obtain 
\begin{align*}
	0\ge \Ric_{\g}(\Y, \Y)&\ge \g(\bnabla \Y, \bnabla \Y)  = \g^{\alpha \beta} \g(\bnabla_{E_\alpha} \Y, \bnabla_{E_\beta} \Y)\\
    &= 2\g^{01} g(\bnabla_{\underline \ell} \Y, \bnabla_{\Y} \Y) + \sum_{a=2}^n \g^{aa} \g(\bnabla_{E_a} \Y, \bnabla_{E_a} \Y) \\
	&=\sum_{a, b=2}^n \g (\bnabla_{E_a} \Y, E_b)^2. 
\end{align*}
If $\Y=0$ at $p$, then $\bnabla \Y=0$ at $p$, and hence $\Y$ vanishes identically, a contradiction.
\end{proof}

\begin{lemma}\label{le:inv}
Let $(\N, \g)$ be a spacetime admitting a causal Killing vector field~$\Y$ and satisfying $\Ric_{\g}(\Y, \Y)\le 0$. Then the set $Z \cap \Int U$ is characterized  as the set of the points $x\in \Int U$ such that, at $x$, 
\begin{align*}
 &\eta = 0, \, \Y\neq 0,  \, \nabla \eta=0, \, L\eta=0, \\
& \Ric_{\g}(\Y, \Y)=0, \quad   X^\perp \mbox{ is involutive},
\end{align*}
where $X^\perp$ denotes the orthogonal distribution to $X$. 
\end{lemma}
\begin{proof}
Because $\eta\ge 0$,   $\eta$ attains a minimum at $x\in Z\cap \Int U$. Hence, $\nabla \eta=0$ and $\nabla^2 \eta$ is nonnegative definite, so $L\eta\ge 0$ at $x$. By \eqref{eq:L},  we  have 
\[
L\eta=0,\quad \g^{ik} (X_{k;j} - X_{j;k}) \g^{j\ell} (X_{i;\ell} - X_{\ell;i})=0,\quad  \Ric_{\g}(\Y, \Y)=0 \quad \mbox{ at } x.
\] By Proposition~\ref{pr:elliptic2}, it follows that $X^\perp$ is involutive at $x$. 
\end{proof}

The proof of Lemma~\ref{Le:smooth} on the regularity of $\partial Z$ follows from the earlier characterization of the null set, together with classical results on degenerate sets of elliptic differential equations reviewed in Appendix~\ref{se:de}.
%\begin{lemma}\label{le:smooth}
%Let  $(\N, \g)$ be a spacetime admitting a causal Killing vector field $\Y$ and satisfying $\Ric_{\g}(\Y, \Y)\le 0$, and let $(U, g, k)$ be an initial data set in $\N$. Let $(f, X)$ be the lapse-shift pair of $\Y$ on $U$ with $f> 0$.   Then $\Pi_p$ is a smooth hypersurface orthogonal to $X$. Consequently, the boundary of the null set, $\partial Z \cap \Int U$, is a smooth  hypersurface in $\Int U$.
%\end{lemma}
\begin{proof}[Proof of Lemma~\ref{Le:smooth}]
Let $z \in \partial Z\cap \Int U$. Recall the propagation set $P_z$ defined in Lemma~\ref{le:pro}:
\[
	P_z = \{ x\in \Int U: x, z \in \gamma(t) \mbox{ where $\gamma(t)$ is an integral curve of $V\in X^\perp$ }\}.
\]
By Lemma~\ref{le:pro}, we have $P_z \subset Z$. By Lemma~\ref{le:inv}, $X^\perp$ is involutive on $P_z$. Therefore, by the Frobenius theorem, $P_z$ is a hypersurface as smooth as $g$ and $X$. Moreover, for any $\hat z \in \partial Z$ sufficiently close to $z$, the hypersurfaces $P_{z}$ and $P_{\hat z}$ must be tangent at their intersection, since both are orthogonal to $X$. Hence $P_{z}$ and $P_{\hat z}$ coincide, which implies that $\partial Z$ coincides with $P_z$ near $z$. 
\end{proof}

We prove the strong maximum principle, Theorem~\ref{Th:strong}, whose statement is recalled below:
\begin{manualtheorem}{\ref{Th:strong}}
Let $(\N,\g)$ be a spacetime admitting a causal Killing vector field $\Y$ with $\Ric_{\g}(\Y,\Y)\le 0$, and let $(U, g, k)$ be an initial data set in $\N$ with smooth boundary $\partial U$ having unit normal $\nu$. Then
\begin{enumerate}
\item Suppose there exists $p \in \partial U$ such that $\eta(x) > \eta(p)$ for all $x \in \Int U$ with $\eta(p) > 0$. Then $\nu( \eta)\neq 0$ at $p$. \label{it:Hopf1}
\item If $\eta>0$ in $\Int U$ and $\eta=0$ in an open subset $\Sigma\subset \partial U$, then $\nu( \eta) \neq 0$ on a dense subset of $\Sigma$.    \label{it:Hopf2}
\item If $\eta$ attains a local minimum at an interior point of $U$, then $\eta$ is constant on $U$.  \label{it:max}
\end{enumerate}
\end{manualtheorem}
\begin{proof}
Item~\eqref{it:Hopf1} follows from the standard Hopf boundary point lemma since \eqref{eq:equ} is uniformly elliptic in this case. 

We now prove Item~\eqref{it:Hopf2}. Suppose, for contradiction, that $\nu(\eta) = 0$, and hence $\nabla \eta = 0$, on an open subset of $\Sigma$,  denoted by~$\Sigma_0$. By Lemma~\ref{Le:smooth}, $\Sigma_0$ is a hypersurface. A collar neighborhood of  $\Sigma_0$, denoted by $N$, is foliated by  level sets $\Sigma_t\cap N$ increasing in $t>0$, with $\nabla \eta\neq 0$ on each $\Sigma_t\cap  N$ for $t>0$ by Item~\eqref{it:Hopf1}. Consequently, the vector field $\nabla \eta - (\nabla \eta \cdot \hat X) \hat X $ is nonzero on each $\Sigma_t\cap  N$ for $t>0$. Consider the normalized vector field:
\[
W = \tfrac{\nabla \eta - (\nabla \eta \cdot \hat X) \hat X}{|\nabla \eta - (\nabla \eta \cdot \hat X) \hat X|}.
\]
Since $W$ is bounded, it extends continuously across $\Sigma_0$. Flowing $\Sigma_t$ along $-W$ back into $\Sigma_0$, we obtain the flowout manifold $K$. Then $K$ satisfies the assumptions in Proposition~\ref{pr:mono}; namely, $K$ has piecewise smooth boundary  $\partial K = \partial_1 K \cup \partial_2 K$, where $\partial_1 K =(\Sigma_0 \cup \Sigma_t)\cap  N$ consists level sets of $\eta$  and $\partial_2 K$ is tangent to $\nabla \eta - (\nabla \eta \cdot \hat X) \hat X$.

Applying the localized monotonicity formula, Proposition~\ref{pr:mono}, we recall, for $t>0$, 
\begin{align*}
 F_K(t) &= \int_{\Sigma_t\cap K }f (\nabla \eta - (\nabla \eta \cdot \hat X)\hat X) \cdot \tfrac{\nabla \eta}{|\nabla \eta|}\, d\sigma \\
 M_K(t) &= t^{-1}F_K(t)
\end{align*} 
and $M_k(t)\ge 0$ is non-increasing in $t>0$. Since $\nabla \eta=0$ on $\Sigma_0$,   we have $\lim_{t\to 0^+} F_K(t)=0$. By Item~\eqref{it:limit} of Proposition~\ref{pr:mono}, 
\begin{align*}
\lim_{t\to 0^+} M_K(t) &\le \lim_{t\to 0^+} \int_{\Sigma_t\cap K}  \left( \tfrac{1}{f} (\nabla \eta)_i  +   2 \hat X^k (X_{k;i} - X_{i;k})\right)  \tfrac{(\nabla \eta)^i}{|\nabla \eta|} \, d\sigma \\
&=\lim_{t\to 0^+} \int_{\Sigma_t\cap K}   2 \hat X^k (X_{k;i} - X_{i;k}) \hat X^i \, d\sigma=0
\end{align*}
where we used that $\Sigma_0$ is orthogonal to $\hat X$. This implies that $M_K(t) \equiv 0$, and thus $\nabla \eta\equiv 0$, a contradiction.

Finally, Item~\eqref{it:max} follows by a standard argument. It suffices to consider the case where the minimum of $\eta$ is zero; otherwise, the standard strong maximum principle applies. If $\eta$ attains this minimum at an interior point $p$, then $Z \cap \Int U \neq \emptyset$. By Lemma~\ref{Le:smooth}, $\partial Z\cap \Int U$ is a smooth hypersurface on which $\nabla \eta \equiv 0$. Applying Item~\eqref{it:Hopf2} to $U\setminus Z$ then yields a contradiction. 
\end{proof}

\section{First Variation and Spacelike Lapse–Shift Pairs} \label{se:spacelike}

Let  $(M, g, k)$ be an asymptotically flat initial data set, possibly with compact boundary, and the ADM energy-momentum $(E, P)$. Recall the modified Einstein constraint operator~$\overline{\Phi}_{(g,k)}$ from~\eqref{eq:varphi0}. Fixing a lapse-shift pair $(f_0, X_0)$ asymptotic to $(E, -P)$, we define the \emph{modified Regge-Teitelboim Hamiltonian} $\mathcal{H} :\mathcal M^{2,\alpha}_{-q} (M)\times \C^{1,\alpha}_{-1-q} (M)\to \mathbb{R} $ by 
\begin{align}\label{equation:functional}
\begin{split}
	\mathcal{H} (\gamma, \tau) &=(n-1)\omega_{n-1} \big(E E(\gamma,\tau) -P \cdot P(\gamma, \tau)\big) \\
	&\quad - \int_M \overline{\Phi}_{(g,k)}(\gamma, \tau) \cdot (f_0, X_0)\, d\mu_g
\end{split}
\end{align}
where recall $\mathcal M^{2,\alpha}_{-q} (M)$ denotes the space of Riemannian metrics $\gamma$ satisfying $\gamma- g_{\mathbb E} \in \C^{2,\alpha}_{-q}$, $(E(\gamma,\tau), P(\gamma, \tau))$ is the ADM energy-momentum of $(\gamma, \tau)$, and both the volume measure $d\mu_g$ and the inner product are taken with respect to~$g$. All results in this section remain valid if  $\overline{\Phi}_{(g,k)}$ is replaced by the $\varphi$-modified operator $\Phi^{\varphi}_{(g,k)}$ for  any bounded smooth function $\varphi$, but we work with  $\overline{\Phi}_{(g,k)}$ for notational simplicity.

The main result of this section, Theorem~\ref{th:causality}, establishes that if $(g, k)$ is a critical point of the above functional with respect to a lapse–shift pair $(f, X)$ that is timelike somewhere, then one can perturb $(g, k)$ to decrease the ADM mass while preserving the dominant energy scalar and, when $\partial M\neq \emptyset$,   the Bartnik boundary data.

We define the \emph{Bartnik boundary data} on $\partial M$ as 
\begin{align}
\label{eq:Bartnik}
 B (g, k) :=\Big (g^\intercal,  H_g, k(\nu_g)^\intercal,  \tr_{g^\intercal} k\Big)
\end{align}
where $g^\intercal$ is the restriction of $g$ on the tangent bundle of $\partial M$, $\nu_g$ is the inward unit normal, $H_g = \Div \nu_g$ is the mean curvature, and $k(\nu_g)^\intercal$ denotes the $1$-form $k(\nu_g, \cdot)$ restricted on the tangent bundle of $\partial M$. More generally, for a symmetric $(0,2)$-tensor $\tau$, we use the notation $\tau (v)^\intercal$ to denote the restriction of $\tau (v, \cdot)$ on the tangent bundle of $\partial M$. 

The linearization of $B$ at $(g, k)$, denoted $B'|_{(g, k)}$, is given by
\[
	B'|_{(g, k)}(h, w) = \Big(h^\intercal, H'|_g(h), w(\nu_g)^\intercal + k(\nu'|_g(h))^\intercal , \tr_{g^\intercal} w - h^\intercal \cdot k^\intercal \Big)
\] 
where $\nu '|_g(h)$ and $H'|_g(h)$ denote the linearization of the unit normal and mean curvature respectively:
\begin{align*}
	\nu '|_g(h)&=-\tfrac{1}{2} h (\nu, \nu) \nu - g^{ab} h(\nu, e_a) e_b \\
	H'|_g(h)&= \tfrac{1}{2} \nu(\tr_g h^\intercal) - \Div_{g^\intercal}h(\nu)^\intercal- \tfrac{1}{2} h(\nu, \nu)H_g
 \end{align*}         
 where $e_a, e_b$ are tangent to $\partial M$. Here and in what follows, we often simply write $\nu$ for $\nu_g$ when the context is clear.

%For the time-symmetric case $k\equiv 0$,  refer to \cite[Proposition 3.7]{Anderson-Khuri:2013}. 
 \begin{lemma} \label{le:first}
The first variation of $\mathcal{H}$ at $(g, k)$ is given by:
\begin{align}\label{eq:first}
\begin{split}
	 &D\mathcal{H}\big|_{(g,k)} (h, w) =- \int_M (h, w) \cdot (D\overline \Phi_{(g,k)})^*(f_0, X_0) \, d\mu_g\\
	& -\tfrac{1}{2}  \int_{\partial M}  \bigg(f_0 A_g - \Big(\nu (f_0) + k(\nu, X_0) \Big)g^\intercal + g(\nu, X_0) k^\intercal , 2f\bigg) \cdot \Big(h^\intercal, H'|_g(h)\Big) \, d\sigma_g \\
	&\small{ + \int_{\partial M} \Big(X_0^\intercal, -g(\nu, X_0) \Big) \cdot \Big(w(\nu)^\intercal + k(\nu'|_g(h))^\intercal, \tr_{g^\intercal} w - h^\intercal \cdot k^\intercal \Big)\, d\sigma_g}
\end{split} 
\end{align}
where $\nu$ is the inward unit normal to $M$. Consequently, if $B'|_{(g,k)}(h, w)=0$, then 
\begin{align*}
	 &D\mathcal{H}\big|_{(g,k)} (h, w) =- \int_M (h, w) \cdot (D\overline \Phi_{(g,k)})^*(f_0, X_0) \, d\mu_g.
\end{align*}

\end{lemma}
\begin{proof}
The interior integral can be derived just as the case where $M$ has no boundary, see, for example, \cite[Theorem 5.2]{Bartnik:2005} (or \cite[Lemma 5.2]{Huang-Lee:2020}), so we focus on deriving the boundary terms.  

Applying the formula $D\overline \Phi_{(g,k)}(h,w)$ from \eqref{eq:phi} and \eqref{eq:phi2} gives
\begin{align*}
	&- \int_M D\overline\Phi_{(g,k)}(h,w) \cdot (f_0, X_0)\, d\mu_g \\
	&= - \int_M (-\Delta_g (\tr_g h) + \Div_g \Div_g h\Big)  f_0\, d\mu_g \\
	&\quad - \int_M\Big( \Div_g w - d(\tr_g w) - g^{jp} k_i^q h_{pq;j} + \tfrac{1}{2} g^{jp} k_i^q h_{pj;q} + \tfrac{1}{2} k^{jq} h_{jq;i} \Big) X_0^i \, d\mu_g\\
	&\quad + \int_M F_0(h,w)  \cdot (f_0, X_0)\, d\mu_g
\end{align*}
where  $F_0(h,w)$ denotes the terms involving only zeroth order in $(h, w)$ and contributes to the interior integral in \eqref{eq:first}.  %Applying integration by parts to the first and second integrals on the right-hand side, the boundary terms on coordinate spheres $S_r$ at $r\to \infty$ cancel with the linearization of the ADM energy and momentum, while the boundary terms on $\partial M$ are exactly those appearing in \eqref{eq:first}. 

Applying integration by parts twice to the first integral gives:
\begin{align*}
&- \int_M (-\Delta_g (\tr_g h) + \Div_g \Div_g h\Big)  f_0\, d\mu_g \\
&= \text{interior integral} - \tfrac{1}{2} E\int_{S_\infty} ( h_{ij,i} - h_{ii,j} ) \nu_0^j \, d\sigma_g\\
&\quad  -\tfrac{1}{2}  \int_{\partial M}\left[ \nu(\tr_g h) f_0 - h_{ij;i} \nu^j f_0 - \nu(f) \tr_g h + h_{ij} (f_0)^j \nu^i\right] \, d\sigma_g\\
&= \text{interior integral} - (n-1)\omega_{n-1} a E'(h,w) \\
& \quad - \tfrac{1}{2} \int_{\partial M} \left( f_0 A_g - \nu(f_0) g^\intercal, 2f_0 \right) \cdot \left(h^\intercal, H'|_g(h)\right)\,d\sigma_g,
%&=\text{bulk terms}- (n-1)\omega_{n-1} a E'|_{(g, k)}(h,w) 
\end{align*}
where $E'$ denotes the linearization of the ADM energy and the last term follows a standard computation, see, e.g. \cite[Proposition 3.7]{Anderson-Khuri:2013}. 

Similarly, integrating by parts twice in the second integral yields:
\begin{align*}
&- \int_M\left( \Div_g w - d(\tr_g w) - g^{jp} k_i^q h_{pq;j} + \tfrac{1}{2} g^{jp} k_i^q h_{pj;q} + \tfrac{1}{2} k^{jq} h_{jq;i} \right) X_0^i \, d\mu_g \\
&= \text{interior integral} + P_i \int_{S_\infty}(w - (\tr_g w) g)(\nu, \partial_i) \, d\sigma_g \\
&\quad + \int_{\partial M} \left[ (w - (\tr_g w) g)(\nu, X_0) - g^{ij} k(X_0, e_i) h(\nu, e_j) \right] \, d\sigma_g\\
&\quad + \int_{\partial M} \left[ \tfrac{1}{2} (\tr_g h) k(\nu, X_0) + \tfrac{1}{2} (h \cdot k) g(\nu, X_0) \right]\, d\sigma_g\\
&= \text{interior integral} + (n-1)\omega_{n-1} P \cdot P'|{(g, k)}(h, w) \\
&\quad + \int_{\partial M}  \left( \tfrac{1}{2} k(\nu, X_0) g^\intercal -\tfrac{1}{2} g(\nu, X_0) k^\intercal \right) \cdot h^\intercal \, d\sigma_g\\
&\quad + \int_{\partial M} \Big(X_0^\intercal, -g(\nu, X_0) \Big) \cdot \Big(w(\nu)^\intercal + k(\nu'|_g(h))^\intercal, \tr_{g^\intercal} w - h^\intercal \cdot k^\intercal \Big)\, d\sigma_g.
\end{align*}
\end{proof}

Next is a computational lemma for timelike $(f, X)$. 
\begin{lemma}\label{le:spacelike}
 Let $(f, X)$ be a lapse-shift pair on an initial data set $(U, g, k)$. Then $f<|X|_g$ on $U$ if and only if there exists a lapse-shift pair $(u, W)$ with $u > |W|_g$  such that 
 \[
    u f + g(W, X) < 0 \quad \mbox{ in } U. 
 \]
\end{lemma}
%\begin{remark}
 %  In a Lorentzian spacetime $(\N, \g)$, the above statement is equivalent to saying that a vector field $\V = f\n + V$ is past-pointing or spacelike on $U$ if and only if there exists a future-directed timelike $\W = u\n + W$  such that $\g( \V, \W)>0$ on~$U$.  
%\end{remark}
\begin{proof}
    Suppose $f<|X|_g$. On the subset $U_0\subset U$ where $|X|_g>0$, we let $u$ be a positive function so that $|X|_g < u < \lambda |X|_g$, where  $\lambda (x)$ is defined by $\lambda(x) = +\infty$ if $f(x)\le 0$ and $\lambda(x) = \frac{|X|_g}{f}$ if $f(x)>0$.  Let $W = -X$. It is straightforward to verify that
    \[
         u f + g(W, X) = u f - |X|_g^2 < 0 \quad \mbox{ on } U_0.
    \]
    We extend $u$ to be any positive function and $W=0$ elsewhere on $U_0$, and it is easy to see that $u f + g(W, X)<0$. 

    Conversely, suppose, for contradiction, that $f\ge |X|_g$. For any $(u, W)$ with $u> |W|_g$, we have 
    \[
        u f + g(W, X) \ge u f - |W|_g |X|_g \ge 0,
    \]
    which leads a contradiction. 
\end{proof}

We need local surjectivity of the modified constraint operator, with the Bartnik boundary data if $\partial M\neq \emptyset$, whose proof is included in Appendix~\ref{se:sur}.

\begin{proposition}\label{pr:sur}
Let $(M, g, k)$ be an asymptotically flat initial data set, and let $\varphi$ be a smooth bounded scalar function on $M$.  Let $m\ge 2$. 
\begin{enumerate}
    \item Suppose $M$ is complete without boundary. Then the smooth map 
    \[
     \overline \Phi_{(g, k)}:\mathcal{M}^{m,\alpha}_{-q}(M)\times \C^{m-1,\alpha}_{-1-q}(M)\to  \C^{m-2, \alpha}_{-2-q}(M)
    \] is locally surjective.
    \item \label{it:bd} Suppose $M$ is complete with  boundary~$\partial M$. Then the smooth map 
    \begin{align*}
    T:\mathcal M^{m,\alpha}_{-q} (M)&\times \C^{m-1,\alpha}_{-1-q}(M)\to \C^{m-2, \alpha}_{-2-q}(M) \times \mathcal B^{m,\alpha}(\partial M)\\
    	&T(g, k) = \Big(  \overline \Phi_{(g, k)}, B(g, k)\Big)
    \end{align*}
    is  locally surjective, where  $\mathcal B^{m,\alpha}(\partial M)$ denotes the corresponding codomain of the Bartnik boundary map $B$. % Specifically, $\mathcal B^{m,\alpha}(\partial M)$ is the product space of $\C^{m,\alpha}(\partial M)\times \C^{m-1,\alpha}(\partial M) \times \C^{m-1,\alpha}(\partial M)$, consisting of symmetric $(0,2)$-tensors on the tangent bundle of $\partial M$,  scalar functions on $\partial M$, and $1$-forms on the tangent bundle of $M$ along $\partial M$.  %Then the linearization of $T$ at $(g, k)$, 
%\[
%	DT(h,w) =  \Big(  D\Phi_{(g, k)} (h,w), B'|_{(g, k)}(h,w)\Big)
%\]
%is surjective. 
\end{enumerate}
\end{proposition}

\begin{theorem}\label{th:causality}
Let $(M, g, k)$ be an asymptotically flat initial data set, possibly with boundary, satisfying the dominant energy condition, with ADM energy-momentum $(E, P)$. Suppose $(f, X)$ is a lapse–shift pair asymptotic to $(E, -P)$ and solves $(D\overline \Phi_{(g, k)})^*(f, X)=0$. If, on some open subset $U\subset M$, we have $f < |X|_g$ and $J = 0$, then there exists a family of asymptotically flat initial data sets $(g(t), k(t))$ satisfying
\begin{enumerate}
\item $(g(0), k(0)) = (g, k)$,
\item $(g(t), k(t))$ satisfies the dominant energy condition for $t\ge 0$,
\item If $\partial M \neq \emptyset$, then $(g(t), k(t))$ has the same Bartnik boundary data, i.e.\ $B(g(t), k(t)) = B(g, k)$,
\item For $t>0$, the ADM energy-momentum $(E(t), P(t))$ satisfies
\[
E(t)^2 - |P(t)|^2 < E^2 - |P|^2.
\]
\end{enumerate}
\end{theorem}
\begin{proof}
By Lemma~\ref{le:spacelike}, there exists $(u, W)$ with $u > |W|_g$ on $   U$ such that 
\begin{align}\label{eq:product}
u f + g(W, X)<0 \quad \mbox{ on }  U.
\end{align}
By shrinking $U$, we may assume that $u-|W|_g\ge c>0$ for some constant $c$ on $U$. Let $\rho$ be a smooth bump function on $ U$ such that $\rho > 0$ in $\Int U$ and $\rho\equiv 0$ outside $U$. By Proposition~\ref{pr:sur}, for $|t|$ sufficiently small, there exists a smooth family of asymptotically flat initial data sets $(g(t), k(t))$ with $(g(0), k(0)) = (g, k)$ such that 
\begin{align}\label{eq:Phi}
     \overline \Phi_{(g, k)}(g(t), k(t)) =  \overline \Phi_{(g, k)}(g, k) + t\rho (u, W),
\end{align}
and,  if $\partial M\neq \emptyset$, the Bartnik boundary data of $(g(t), k(t))$ agree with those of $(g, k)$ on $\partial M$. 

Denote by $(\mu_t, J_t)$ and $(\mu, J)$ the energy and momentum densities of $(g(t), k(t))$ and $(g, k)$, respectively. By Lemma~\ref{lemma:sigma_preserve}, we have 
\[
    \mu_t - |J_t|_{g(t)} \ge \mu - |J|_g + t\rho \left(u -|W|_g -  \tfrac{1}{2} |g(t) - g|_g |W|_g\right).
\]
For $t>0$ sufficiently small, we have $u -|W|_g > \tfrac{1}{2} |g(t) - g|_g |W|_g$. Thus, $(g(t), k(t))$ satisfies the dominant energy condition. 

Since $(D\overline \Phi_{(g, k)})^*(f, X)=0$ and $(g(t), k(t))$ has the same Bartnik boundary data, by Lemma~\ref{le:first}, letting $\mathcal H$ be the Regge-Teitelboim Hamiltonian associated to $(f, X)$, we have
\begin{align*}
\begin{split}
	0&=\left. \tfrac{d}{dt} \right|_{t=0}\mathcal{H}(g(t), k(t)) \\
&=(n-1)\omega_{n-1} \left. \tfrac{d}{dt} \right|_{t=0} (E(t)^2 - |P(t)|^2 ) - \int_U \rho(u, W)\cdot (f, X)\, d\mu_g\\
    & >  (n-1)\omega_{n-1} \left. \tfrac{d}{dt} \right|_{t=0} (E(t)^2 - |P(t)|^2 ),
\end{split}
\end{align*}
where we use the definition \eqref{equation:functional}  of $\mathcal H$ and~\eqref{eq:Phi}  in the second line, and  \eqref{eq:product} in the last line. It completes the proof.  
\end{proof}

\section{Proofs of Theorems~\ref{Th:equality} and \ref{Th:Bartnik}}  \label{se:min}

 We present a uniform approach to the equality case of the positive mass theorem, Theorem~\ref{Th:equality}, and the Bartnik stationary vacuum conjecture, Theorem~\ref{Th:Bartnik}.

To set the stage, let $(M, g, k)$ be an asymptotically flat initial data set of type $(\alpha, q, \epsilon)$ satisfying the dominant energy  condition and having the ADM mass-momentum $E\ge |P|$. We assume that $M$ is complete, either with or without boundary. Let $\mathcal U$ be a sufficiently small neighborhood of $(g, k) $ in $\mathcal M^{2,\alpha}_{-q}(M)\times \C^{1,\alpha}_{-1-q}(M)$ consists of asymptotically flat initial data sets of type $(q, \alpha, \epsilon)$. 

\begin{itemize}
\item 
If $M$ has no boundary, we denote the constraint set 
\[
	\mathscr C= \big\{ (\gamma, \tau) \in \mathcal U:  (\gamma, \tau) \mbox{ satisfies the dominant energy condition} \mbox{ in $M$}\big\}.
\]
\item 
If $M$ has nonempty boundary, we denote the constraint set 
\begin{align*}
		\mathscr  C_{B}= \big\{ (\gamma, \tau) \in \mathcal U&:  (\gamma, \tau) \mbox{ satisfies the dominant energy condition} \mbox{ in $M$} \\
		&\qquad\qquad \mbox{and } B(\gamma, \tau) = B(g, k)\big\},
\end{align*}
\end{itemize}
%\margin{Is it more natural to assume that $\gamma^\intercal=g^\intercal$ and that $\tilde B(\gamma,\tau)-\tilde B(g,k)$ is a null vector (as opposed to $\tilde B(\gamma,\tau)-\tilde B(g,k)=0$)? Here $\tilde B(\gamma,\tau)= (  H_\gamma,  \tau(\nu_\gamma, \cdot),  \tr_{\gamma^\intercal} \tau ) $. This suffices to apply the PMT with corners.}
where  $B(\gamma,\tau)$ denotes the Bartnik boundary data from \eqref{eq:Bartnik}.

We say that $(g, k)$ \emph{minimizes the ADM mass} among $\mathscr C$ (respectively, $\mathscr C_B$)  if  for all  $(\gamma, \tau)\in \mathscr C$ (respectively, $(\gamma, \tau)\in \mathscr C_B$), we have 
\[
  E(\gamma, \tau) \ge |P(\gamma, \tau)| \quad \mbox{ and } \quad  E(\gamma, \tau)^2 - |P(\gamma, \tau)|^2 \ge E^2 - |P|^2.
\]
By applying Bartnik's method of Lagrange multipliers and the $\varphi$-modified constraint operator introduced in ~\cite{Huang-Lee:2024}, we obtain the following result. The proof combines established results with appropriate modifications.
\begin{theorem}\label{th:impro}
Let $(M, g, k)$ be an asymptotically flat initial data set of type $(q, \alpha, \epsilon)$ satisfying the dominant energy condition, with $E \ge |P|$ and $E>0$. Suppose $(g,k) \in \C^5_{\mathrm{loc}}(\Int M)\times \C^4_{\mathrm{loc}}(\Int M)$. Assume either:
\begin{enumerate}
\item \label{it:nobdry}$M$ is complete without boundary and $(g, k)$ minimizes the ADM mass among $\mathscr C$, or
\item $M$ is complete with boundary and $(g, k)$ minimizes the ADM mass among $\mathscr C_B$. \label{it:bdry}
\end{enumerate}
Then there exists a nontrivial lapse-shift pair $(f, X)$ with asymptotics $(f, X) - (E, -P) \in \C^{2,\alpha}_{-q}(M)$ satisfying
\begin{align}
(D\overline{\Phi}_{(g,k)})^*(f,X)&=0,\label{eq:pair}\\ 
fJ + |J|_g X&=0.\label{eq:J}
\end{align}

Consequently, $(M, g, k)$ embeds into a spacetime $(\N, \g)$ admitting a Killing vector field $\Y$, with $(f, X)$ as the lapse–shift pair.  The Einstein tensor of $\g$ satisfies $G = \frac{|J|_g}{f^2} \Y \otimes \Y$, where $\Y$ null wherever $J \neq 0$; in particular $\Ric_{\g}(\Y, \Y)=0$ and $G= O^{0,\alpha}(|x|^{-n-\epsilon})$.  

Furthermore, combining Theorem~\ref{th:causality}, $\Y$ must be causal. 
\end{theorem}
\begin{remark}\label{re:reg}
The assumption $(g,k) \in \C^5_{\mathrm{loc}}(\Int M)\times \C^4_{\mathrm{loc}}(\Int M)$ is only required for the application of Theorem 4.1 in \cite{Huang-Lee:2024}.
\end{remark}

\begin{proof}
We focus on proving the existence of such $(f, X)$. The remaining statement about $(\N, \g)$ follows from Theorem 6 in \cite{Huang-Lee:2024}.

We adapt Bartnik’s method of Lagrange multipliers, as implemented in \cite{Huang-Lee:2020}. We demonstrate the argument for case \eqref{it:bdry}, as the case \eqref{it:nobdry} without boundary is strictly simpler.  Fix a lapse-shift pair $(f_0, X_0) -  (E, -P) \in \C^{\infty}_{-q} (M)$. Recall the modified Regge–Teitelboim Hamiltonian~\ref{equation:functional} with respect to $(f_0, X_0)$:
\begin{align*}
\begin{split}
	\mathcal{H} (\gamma, \tau) &=(n-1)\omega_{n-1} \big(E E(\gamma,\tau) - P \cdot P(\gamma, \tau)\big) \\
	&\quad - \int_M \overline \Phi_{(g,k)}(\gamma, \tau) \cdot (f_0, X_0)\, d\mu_g.
\end{split}
\end{align*}

Let $\mathcal U$ be a sufficiently small open neighborhood of $(g, k)$ in the $\C^{2,\alpha}_{-q}(M)\times \C^{1,\alpha}_{-1-q}(M)$ norm. For $(\gamma, \tau)\in \mathcal U$ with $\overline\Phi_{(g,k)}(\gamma, \tau) = \overline \Phi_{(g,k)}(g, k) $ and the same Bartnik boundary data as $(g, k)$, Lemma~\ref{lemma:sigma_preserve}  implies that $(\gamma, \tau)$ satisfies the dominant energy condition and thus $(\gamma, \tau) \in \mathscr C_B$. The assumption that $(g, k)$ minimizes the ADM mass among $\mathscr C_B$, together with the reverse Cauchy-Schwarz inequality for future-directed causal vectors $(E, P)$ and $(E(\gamma, \tau), P(\gamma, \tau))$,  implies that $(g, k)$ minimizes $\mathcal H$ among all such $(\gamma, \tau)$, 
\[
	\mathcal H (\gamma, \tau) \ge \mathcal H(g, k).
\] 
    
To apply the method of Lagrange multipliers of Bartnik~\cite{Bartnik:2005} (see also \cite[Theorem D.1]{Huang-Lee:2020}), we use that the constraint of the minimization problem
\[
(h, w) \to (D\overline \Phi_{(g,k)}(h, w), B'|_{(g,k)}(h,w) ) \in \C^{0,\alpha}_{-2-q}(M)\times \mathcal B^{2,\alpha}(\partial M)
\] 
 is surjective by Proposition~\ref{pr:sur}. Therefore, there exists a Lagrange multiplier $(f_1, X_1, \beta)$ in the dual space of the constraint, that is,  $(f_1, X_1)\in (\C^{0,\alpha}_{-2-q}(M))^*$ and $\beta\in (\mathcal B^{2,\alpha}(\partial M))^*$
 %. In particular, $(\tilde f, \tilde X, \beta)$ lies in the dual space of compactly supported smooth functions and thus defines a distribution, 
 such that
 \[
 D\mathcal H|_{(g, k)} (h,w) = (f_1, X_1) \big(D\overline \Phi_{(g, k)}(h, w)\big) + \beta \big( B'|_{(g, k)}(h,w)\big)
\]	
for all  $(h, w) \in \C^{2,\alpha}_{-q}(M)\times \C^{1,\alpha}_{-1-q}(M)$. Note that  in the right hand side, $( f_1, X_1)$ and $\beta$ are bounded linear functionals acting on  the corresponding spaces. 
Taking $(h, w)$ to be compactly supported and using the  first variation formula \eqref{eq:first}  yields that $(f_1, X_1)$ weakly solves 
%\[
%	 D\mathcal F |_{(g, k)} (h,w)  =(\tilde f, \tilde X) \big( D\overline \Phi_{(g, k)}(h, w)\big).
%\]
%Combined with the first variation formula \eqref{eq:first} for $D\mathcal F |_{(g, k)} (h,w)  $ implies that $(\tilde f, \tilde X)$ is a weak solution (as a distribution) to the equation
\[
	(D\overline \Phi_{(g, k)})^* (f_1, X_1)= -  (D\overline \Phi_{(g, k)})^*(f_0, X_0) \in  \C^{0,\alpha}_{-2-q} (M)\times \C^{1,\alpha}_{-1-q}(M).
\]	
By elliptic regularity, $ (f_1, X_1) \in \C^{2,\alpha}_{\mathrm{loc}}(M) $.   Lemma~\ref{le:asy} on the asymptotics of $(f_1, X_1)$, together with the fact that $(f_1, X_1)$ defines a bounded linear functional,  implies that $(f_1, X_1) \in  \C^{2,\alpha}_{-q}$.  Setting $(f, X) =  (f_0, X_0) + (f_1, X_1)$ gives \eqref{eq:pair} with desired asymptotics $(f, X)-(E, -P)\in \C^{2,\alpha}_{-q}(M)$. 

The above arguments apply equally well for the $\varphi$-modified operators. Thus, for every bounded smooth $\varphi$,  there is a lapse-shift pair $(f, X)$ asymptotics to $(E, -P)$ solving  $(D\Phi^{\varphi}_{(g, k)})^*(f, X)=0$. By Theorem 4.1 of \cite{Huang-Lee:2024}, this implies that, for all $\varphi$ in a dense subset of smooth functions, the solution $(f, X)$ satisfies $(D\overline{\Phi}_{(g, k)})^*(f, X) = (D\Phi^{\varphi}_{(g, k)})^*(f, X) = 0$, as well as \eqref{eq:J}. This completes the proof.

%The only missing property is the $J$-null vector equation \eqref{eq:J}. To achieve this, note that all of the arguments carry over verbatim if we replace the modified constraint operator with the $\varphi$-modified constraint operator $\Phi_{(g,k)}^\varphi$  for any smooth bounded $\varphi$, provided $\varphi$ is sufficiently close to zero so that Lemma~\ref{lemma:sigma_preserve} applies. The same reasoning then implies that for each such $\varphi$, there exists a pair $(f, X)$ with asymptotics $(f, X)-(E, -P)\in \C^{2,\alpha}_{-q}(M)\times \C^{1,\alpha}_{-1-q}(M)$ solving the corresponding adjoint equation for the $\varphi$-modified constraint operator: $(D \Phi^\varphi_{(g, k)})^* (f, X) = 0$. By Theorem 4.1 of \cite{Huang-Lee:2024}, this implies that $(f, X)$ must also satisfy  the $J$-null vector equation \eqref{eq:J}. It completes the proof.
\end{proof}

We prove the equality case of the positive mass theorem, Theorem~\ref{Th:equality}. 

%\begin{theorem}\label{th:pmt}
%Let $(M, g, k)$ be an asymptotically flat initial data set, without boundary, satisfying  the dominant energy condition and with $E=|P|$. Suppose  $(g,k) \in \C^5_{\mathrm{loc}}(\Int M)\times \C^4_{\mathrm{loc}}(\Int M)$.  Suppose that $(g, k)$ minimizes the ADM mass among $\mathscr C$. Then $(M, g, k)$ is embedded in a pp-wave  $(\N, \g, \Y)$ where the Einstein tensor  of $\g$ satisfies $G = \rho \Y \otimes \Y$ for some $\rho\ge 0$.
%\end{theorem}
\begin{proof}[Proof of Theorem~\ref{Th:equality}]
Let $(M, g, k)$ be a complete, boundaryless, asymptotically flat initial data set with the dominant energy condition and the ADM energy-momentum $E=|P|$. Since the case $E=0$ is known (see, e.g. \cite{Huang-Lee:2025} and the references therein), we assume $E > 0$. By Theorem~\ref{th:impro},  $(M, g, k)$ embeds in $(\N, \g)$ with a causal Killing vector field $\Y$ satisfying $\Ric_{\g}(\Y, \Y)=0$.  By the strong maximum principle, Theorem \ref{Th:strong}, either $\Y$ is timelike or null everywhere on $\Int M$. However, since  $(f, X)$ is asymptotic to a null vector $(E, -P)$, $\Y$ cannot be timelike by Theorem~\ref{Th:uni}. Thus,  $\Y$ is a null Killing vector field. Then we apply Theorem~\ref{Th:pp} (prove in Section~\ref{S:vector} below) to conclude that $(\N, \g)$ is a pp-wave.
\end{proof}

Next, we define \emph{admissible extensions} in the context of Bartnik’s quasi-local mass.
\begin{definition}\label{definition:admissible}
Let $(\Omega, g_0, k_0)$ be a $n$-dimensional compact initial data set with nonempty smooth boundary.
We say that $(M, g, k)$ is an \emph{admissible extension of $(\Omega, g_0, k_0)$} if the following holds:
\begin{enumerate}
\item $(M, g, k)$ is an $n$-dimensional asymptotically flat initial data set  with boundary $\partial M$, satisfying the dominant energy condition. 
\item  There exists an identification of the boundaries $\partial M$ and  $\Sigma:=\partial \Omega$ via diffeomorphism, and under this identification, the Bartnik boundary data $B(g, k) = B(g_0, k_0)$  along~$\Sigma$.
\item $(M, g, k)$ satisfies condition~$(\mathscr{C}_1)$ that $\partial M$ is strictly outward-minimizing; that is, it has strictly less $(n-1)$-volume than any homologous hypersurface.
\end{enumerate}
\end{definition}
\begin{remark}
Our proof to the Bartnik stationary vacuum conjecture, which is implied by  Theorem~\ref{Th:Bartnik} below, does not rely explicitly on condition~$(\mathscr{C}_1)$, but only on the fact that it is an open condition under deformations of initial data in $\C^{2}_{-q}(M) \times \C^{1}_{-1-q}(M)$ that preserve the Bartnik boundary data.  Specifically, if $(M, g, k)$ is an admissible extension, then any $(\gamma, \tau)\in \mathscr C_B$ that is sufficiently close to $(g, k)$ is also an admissible extension.  

Two alternative conditions can replace $(\mathscr{C}_1)$: 
\begin{itemize}
\item Condition~$(\mathscr{C}_2)$: There is no MOTS that is an outer embedded boundary, except possibly $\partial M$ itself.
\item Condition~$(\mathscr{C}_3)$: A strengthening of $(\mathscr{C}_1)$,  requiring that $\partial M$ is not a MOTS. 
\end{itemize}
It is shown in \cite[Proposition 7.5]{Huang-Lee:2024} that the condition $(\mathscr{C}_2)$ is an open condition for $3\le n < 7$. 
\end{remark}

%The \emph{Bartnik mass} is defined as%
%\[
%m_B(\Omega, g_0, k_0) =\inf\big \{ m_{\mathrm {ADM}} (M, g, k): \mbox{admissible extension } (M, g, k)\big\}.
%\]
%An admissible extension $(M, g, k)$ whose ADM mass realizes the infimum is called a \emph{Bartnik mass minimizer} for $(\Omega, g_0, k_0)$.

\begin{proof}[Proof of Theorem~\ref{Th:Bartnik}]
Let $(M, g, k)$ be an asymptotically flat initial data set with boundary with the dominant energy condition. Suppose it has ADM energy-momentum $E>|P|$ and minimizes the ADM mass among $\mathscr C_B$. Theorem~\ref{th:impro} implies that $(M, g, k)$ embeds in a spacetime $(\N, \g)$ with a causal Killing vector field $\Y  $, whose lapse-shift $(f, X)$  is asymptotic to $(E, -P)$,  and  such that $\Ric_{\g}(\Y, \Y)=0$ and $G= O^{0,\alpha}(|x|^{-n-\epsilon})$.  By the strong maximum principle (Theorem \ref{Th:strong}), $\Y$ is either timelike or null everywhere on $\Int M$. Since  $E>|P|$, $\Y$ is timelike on $\Int M$.  
\end{proof}

\section{Null Killing Vectors}\label{S:vector}
We prove Theorem \ref{Th:pp} in this section.  For next two lemmas, it is  convenient to use a local \emph{double null frame} $\{ E_0, E_1, E_2, \dots, E_n\}$ on $\N$, where $E_0 = \Y$,  $E_1 = \underline \ell$ is the conjugate null vector such that $\g(\Y, \underline \ell)=-1$, and $\{ E_2, \dots, E_{n}\}$ is a local orthonormal set orthogonal to both $\Y$ and $\underline \ell$. Denote $\Y^\perp = \Span\{ \Y, E_2, \dots, E_n\}$. The following lemma is a direct consequence of Lemma~\ref{le:Y}.

\begin{lemma}\label{le:zero}
Suppose $(\N, \g)$ is a spacetime admitting a null Killing vector field~$\Y$ and satisfying $\Ric_{\g}(\Y, \Y)\le 0$. Then $\Ric_{\g}(\Y, \Y)\equiv 0$, $\Y$ does not vanish,   
\begin{align}\label{eq:Yvanish}
	\Y_{b| a} := \g(\bnabla_{E_a} \Y, E_b)=0 \quad \mbox{ for } a, b = 2, \dots, n,
\end{align}
and 
\begin{align}\label{eq:Yall2}
\left\{ 
\begin{array}{ll}
\bnabla_{E_a} \Y &= - \g(\bnabla_{E_a} \Y, \underline \ell) \Y \\
 \bnabla_{\Y} \Y &=0  \\
 \bnabla_{\underline \ell } \Y &   =- \g(\bnabla_{E_a } \Y, \underline \ell) E_a
\end{array}\right..
 \end{align}
\end{lemma}

\begin{lemma}\label{le:integral}
Let $(\N, \g)$ be a spacetime with a null Killing vector field~$\Y$, and let  $(U, g, k)$ be an initial data set with  the lapse-shift pair $(f, X)$ of $\Y$ along $U$. If $\Ric_{\g}(\Y, \Y)\le 0$, then  $\Y^\perp$ is involutive. Let $S$ be an integral hypersurface of $\Y^\perp$. Then $\Sigma= U \cap S$ is orthogonal to the vector field $X$ and is a MOTS in $U$; that is, $H_g + \tr_{g^\intercal} k=0$, where $H_g = \Div_g \hat X$ and $\hat X = \frac{X}{f}$.
\end{lemma}
\begin{proof}

Using the fact that $\Y$ is Killing and  \eqref{eq:Yvanish}, we see that $\Y^\perp$ is involutive, and hence admits an integral hypersurface $S$  by the Frobenius theorem.  Denote by $\n$ the future-directed timelike unit normal to $U$ and consider a local orthonormal frame $\{ E_2, \dots, E_n\}$ of $\Sigma$. Since $E_a$ is orthogonal to both $\n$ and $\Y$ for all $a=2, \dots, n$, we have $ g(E_a, X) = \g(E_a, \Y) - \g(E_a, f\n) = 0$; thus, $\hat X$ is orthogonal to $\Sigma$. The mean curvature with respect to $\hat X$ is given by 
\begin{align*}
	H_g = \Div_g \hat X = f^{-1} \Div_g X - \hat X\cdot \tfrac{\nabla f}{f} = - \tr_g k + k(\hat X, \hat X) = - \tr_{g^{\intercal}} k,
\end{align*}
where we use $\mathscr L_X g = -2k f$ (which follows from $\mathscr L_{\Y} \g=0$ restricted to $U$) and $\nabla (f^2- |X|_g^2)=0$. 
\end{proof}

By \eqref{eq:Yall2}, $\bnabla \Y=0$ if and only if $\g(\bnabla_{E_a} \Y, \n)=0$ for all $E_a$ tangent to $\Sigma$. This motivates us to define the 1-form 
\[
\omega (\V) = \g(\bnabla_{\V} \Y, \n) 
\]
and to identify the conditions under which $\omega =0$. Note that along an initial data set $(U, g, k)$, we have
\begin{align}\label{eq:omega}
	\omega = -df - k(X,\cdot).
\end{align}

The following proposition shows that $\omega$ is a closed 1-form when restricted to $\Sigma$ and satisfies a key differential equation. In the next proposition, we let  $\{ e_0, e_1, e_2, \dots, e_n\}$ be a local orthonormal frame, where   $e_0= \n$, $e_1 = \hat X$, and $\{ e_2, \dots, e_n\}$ are tangent to $\Sigma$, 
\begin{proposition}\label{pr:omega}
Let $(\N, \g)$ be a spacetime admitting a null Killing vector field~$\Y$. Suppose that $\Ric_{\g}(\Y, \Y)=0$ everywhere.  Let $\Sigma$ be an $(n-1)$-dimensional spacelike submanifold as defined in Lemma~\ref{le:integral}. Then 
\begin{enumerate}
\item We have $\omega(\n)=0$ and $\omega(X)=0$, so the vector field dual to $\omega$, denoted by $W$,  is tangent to $\Sigma$.  
\item $\omega$ is a closed 1-form on $\Sigma$. 
\item $\omega$ and the dual vector field $W$ satisfy the following differential equations:
\begin{align}
 	\Div_{\g} \omega & = \Div_{\g} W= - \Ric_{\g} (\Y, \n) - \g^{\alpha \beta} \g(\bnabla_{e_\alpha} \Y,\bnabla_{e_\beta} \n)\quad \mbox{ on } \N \notag \\
	\Div_g  W &= - \Ric_{\g} (\Y, \n) \quad \mbox{ on } M \label{eq:div-M}\\
	\Div_\Sigma W &= g(\nabla_{\hat X} \hat X, W)  - \Ric_{\g} (\Y, \n) \quad \mbox{ on } \Sigma. \label{eq:div-sigma}
\end{align}  
\end{enumerate}  
\end{proposition}
\begin{proof}
The first item follows immediately from $\g(\bnabla_X \Y, \n) = \g(\bnabla_{\Y} \Y, \n) - \g(\bnabla_{f\n} \Y, \n) =0$.

 We next show that $\omega$ is a closed 1-form on~$\Sigma$. We compute, for $a, b\in \{ 2, \dots, n\}$, 
\begin{align*}
		&(\nabla^\Sigma _{e_b} \omega)(e_{a}) - (\nabla^\Sigma_{e_a} \omega)(e_b) = \omega_{a;b} - \omega_{b;a} \\
		&= \Y_{0| ab} + \g(\bnabla_{e_a} \Y, \bnabla_{e_b} \n) - \Y_{0|ba} - \g(\bnabla_{e_b} \Y, \bnabla_{e_a} \n) \\
		&= \mathbf R_{a0b\delta} \Y^\delta - \mathbf R_{b0a\delta} \Y^\delta + \g(\bnabla_{e_a} \Y, \bnabla_{e_b} \n) - \g(\bnabla_{e_b} \Y, \bnabla_{e_a} \n) \\
		&= -\mathbf R_{0ba\delta} \Y^\delta - \mathbf R_{ba0\delta} \Y^\delta- \mathbf R_{b0a\delta} \Y^\delta + \g(\bnabla_{e_a} \Y, \bnabla_{e_b} \n) - \g(\bnabla_{e_b} \Y, \bnabla_{e_a} \n)\\
		&= - \mathbf R_{ba0\delta} \Y^\delta + \g(\bnabla_{e_a} \Y, \bnabla_{e_b} \n) - \g(\bnabla_{e_b} \Y, \bnabla_{e_a} \n)\\
		&= \Y_{a|b0} + \g(\bnabla_{e_a} \Y, \bnabla_{e_b} \n) - \g(\bnabla_{e_b} \Y, \bnabla_{e_a} \n) = e_0 (\Y_{a|b} ) = 0
\end{align*}	  
where we use \eqref{eq:DDY} in the third line, the Bianchi identity in the fourth line, and $\Y_{a|b}\equiv 0$ by \eqref{eq:Yvanish} in the last line.

We compute 
\begin{align*}
	\Div_{\g} W &=	\Div_{\g} \omega =\g^{\alpha\beta} \omega_{\alpha |\beta}= \g^{\alpha \beta} \big( \Y_{0| \alpha \beta} -  \g(\bnabla_{e_\alpha} \Y,\bnabla_{e_\beta} \n)\big)\\
	& =- \Ric_{\g} (\Y, \n) - \g^{\alpha \beta} \g(\bnabla_{e_\alpha} \Y,\bnabla_{e_\beta} \n).
\end{align*}
By Lemma~\ref{le:div}, the above left hand side equals
\[
\Div_{\g} W  = \Div_g W + g( f^{-1}\nabla f, W) = \Div_g W + \g (\bnabla_{f^{-1} \nabla f} \Y, \n). 
\]
On the other hand, using \eqref{eq:n} that  $\bnabla_{\n} \n = f^{-1} \nabla f$, we have 
\begin{align*}
 - \g^{\alpha \beta} \g(\bnabla_{e_\alpha} \Y,\bnabla_{e_\beta} \n)&= - \g(\bnabla_{\n} \Y, \bnabla_{\n} \n) -  \g(\bnabla_{e_j} \Y, \bnabla_{e_j} \n)\\
 &=- \g(\bnabla_{\n} \Y, f^{-1} \nabla f) + k_{ij} \g(\nabla_{e_j} \Y, e_i) \\
 &= \g(\bnabla_{f^{-1} \nabla f} \Y, \n ) 
\end{align*}
where we use that $k_{ij}$ is symmetric and $\g(\nabla_{e_j} \Y, e_i)$ is anti-symmetric. Combining the above computations gives \eqref{eq:div-M}. Then \eqref{eq:div-sigma} follows directly. 
\end{proof}

Combining the above results with the assumption of asymptotic flatness, we prove the first part of Theorem \ref{Th:pp}.

\begin{proposition}
Let $(\N, \g)$  be a spacetime admitting a null Killing vector field~$\Y$  and an asymptotically flat initial data set $(M, g, k)$ where $M$ is complete without boundary. Assume  $\Ric_{\g}(\Y, \Y) = 0$ on $\N$, the associated lapse-shift pair $(f, X)$ is asymptotically translational on $M$, and $\Ric_{\g}(\Y, \n) = 0$ along $M$, where $\n$ is the unit normal to $M$.  
Then $\Y$ is parallel in the domain of dependence of $M$ in $(\N, \g)$. 
\end{proposition}

\begin{proof}

 Since $(f, X)$ is asymptotically translational and each leaf of the foliation $\{\Sigma\}$ of $M$ is orthogonal to $X$,  leaves outside a sufficiently large compact set of $M$ are diffeomorphic to $\mathbb R^{n-1}$, and Reeb’s stability theorem then implies all leaves are $\mathbb R^{n-1}$.
  
Proposition~\ref{pr:omega} gives that $\omega$ is a closed 1-form on $\Sigma$, so $\omega = dz$  for some function $z$. By \eqref{eq:omega}, $\omega \to 0$ at infinity, so we can take $z\to 0$ at infinity.      Applying \eqref{eq:div-sigma} for $W = \nabla^\Sigma z$ gives 
\[
\Delta_\Sigma z  -g(  \nabla_{\hat X} \hat X, \nabla^\Sigma z) = 0.
\]
The maximum principle then yields $z\equiv0$, hence $\omega \equiv 0$ on $M$, and $\omega=0$ along any integral curve of $\Y$.
\end{proof}

%\begin{lemma}
%    The level-sets $\Sigma$ are asymptotically flat.
%\end{lemma}

%\begin{proof}
%Let $\{e_i\}$ be an asymptotically flat coordinate system such that $|X|^{-1}X$ asymptotes to $e_n$.
%Since $(f,X)$ is asymptotically translational to some $(a,b)\in\mathbb R\times\mathbb R^n$ with 
%$$
%(f,X)=(a,b)+O^{2,\alpha}(|x|^{-q}),
%$$
%we have
%$$
%|X|^{-1}X=e_n+O^{2,\alpha}(|x|^{-q}).
%$$
%Let $\tilde e_i=e_i-|X|^{-2}\langle e_i,X\rangle X$.
%Then the induced metric $g_\Sigma$ satisfies
%\begin{align}\label{eq Sigma is AF}
%\begin{split}
%    g_\Sigma(\tilde e_i,\tilde e_j)=&g(\tilde e_i,\tilde e_j)\\
%    =&g(e_i-|X|^{-2}\langle e_i,X\rangle X,e_j-|X|^{-2}\langle e_j,X\rangle X)\\
 %   =&\delta_{ij}-|X|^{-2}g(e_i,X)g(e_j,X)+O^{2,\alpha}(|x|^{-q})\\
%    =&\delta_{ij}-g(e_i,e_n)g(e_j,e_n)+O^{2,\alpha}(|x|^{-q})\\
%    =&\delta_{ij}+O^{2,\alpha}(|x|^{-q})
%    \end{split}
%\end{align}
%Therefore, $\{\tilde e_i\}$ form a $\C^{2,\alpha}_{-q}$-asymptotically flat coordinate system for $\Sigma$, or in other words $\Sigma$ is asymptotically flat of order $q$.
%\end{proof}

 A triple $(\N, \g, \Y)$ is called a \emph{Brinkmann spacetime} if $(\N, \g)$ is a spacetime admitting a  parallel null vector field~$\Y$. After deriving some general properties of Brinkmann spacetimes, we prove Proposition~\ref{pr:Brinkmann} that implies the second part of Theorem~\ref{Th:pp}.
\begin{lemma}\label{le:sigma}
Let $(\N, \g, \Y)$ be a Brinkmann spacetime and $S$ an integral hypersurface of $\Y^\perp$. Let $\Sigma$ be any spacelike $(n-1)$-dimensional submanifold of $S$ with the induced metric $h$. Then 
\[
	\Ric_{\g}(V, W) = \Ric_h(V, W)
\]
for any vector fields $V, W$ tangent to $\Sigma$. 
\end{lemma}
\begin{proof}
Consider a double null frame $\{ E_0, E_1, E_2, \dots, E_n\}$  where $E_0=\Y$, $E_1=\underline \ell$ the conjugate null vector,  and $\{ E_2, \dots, E_n\}$ is a local orthonormal frame on $\Sigma$. Let the indices $a, b, c, d$ range over $2, \dots, n$ below. We compute the second fundamental form of $\Sigma\subset \N$ and use that $\Y$ is parallel to obtain 
\[
	(\bnabla_{E_a} E_b)^\perp = -\g(\bnabla_{E_a} E_b, \underline \ell) \Y  - \g(\bnabla_{E_a} E_b, \Y) \underline \ell = -\g(\bnabla_{E_a} E_b, \underline \ell) \Y.
\]
Together with the Gauss equation,
\begin{align*}
	\mathrm{Rm}_{\g}(E_a, E_b, E_c, E_d) &= \mathrm{Rm}_h(E_a, E_b, E_c, E_d) \\
	&\quad - \g( (\bnabla_{E_a}E_d)^\perp, (\bnabla_{E_b} E_c)^\perp) + \g( (\bnabla_{E_a} E_c)^\perp, (\bnabla_{E_b} E_a)^\perp)\\
	&=  \mathrm{Rm}_h(E_a, E_b, E_c, E_d).
\end{align*}
Therefore, 
 \begin{align*}
  0 &= \Ric_{\g} (E_a, E_d)=\g^{\beta\gamma} \mathrm{Rm}_{\g}(E_a, E_\beta, E_\gamma, E_d) \\
  &=  \g^{01} \mathrm{Rm}_{\g}(E_a, \Y, \underline \ell, E_d) +\g^{10} \mathrm{Rm}_{\g}(E_a, \underline \ell, \Y, E_d)  +\mathrm{Rm}_{\g}(E_a, E_b, E_b, E_d) \\
  &= \Ric_h (E_a, E_d),
 \end{align*}
 where we use that $\g^{00} = \g^{11} = 0$, and that the curvature tensor components involving $\Y$ vanish because $\Y$ is parallel.

\end{proof}

\begin{proposition}\label{pr:Brinkmann}
Let $(\N, \g, \Y)$ be a  Brinkmann spacetime satisfying $\Ric_{\g}(\V, \W) = 0$ for all spacelike vector fields $\V, \W \in \Y^\perp$. Then the induced metric on any integral hypersurface $S$ of $\Y^\perp$ is Ricci-flat. Moreover, if either $n=3,4$ or $n\ge 5$ and $S$ contains a complete, boundaryless, asymptotically flat  $(n-1)$-dimensional spacelike submanifold, then  the induced metric on $S$ is flat. %In particular, when $n=3, 4$, the spacetime must be Minkowski.%\footnote{Note that an asymptotically flat Ricci flat manifold is Riemann flat.}  
\end{proposition}

\begin{proof}
Since $S$ is a null hypersurface, showing that $S$ is Ricci flat reduces showing that any $(n-1)$-dimensional spacelike submanifold  $\Sigma$ is Ricci flat. By  Lemma~\ref{le:sigma} and the assumption  $\Ric_{\g} (\V, \W)=0$ for all spacelike $\V, \W \in \Y^\perp$, the induced metric on $\Sigma$ is Ricci flat.

If $n=3$ or $4$, $\Sigma$ has dimension $2$ or $3$, and is therefore flat. For $n \ge 5$, a complete, boundaryless, asymptotically flat, Ricci-flat manifold is isometric to Euclidean space (e.g., by the Bishop–Gromov volume comparison theorem). Hence, assuming $\Sigma$ is asymptotically flat, it is Euclidean, so the induced metric on $S$ is flat along $\Sigma$. Since curvature is constant along the integral curves of $\Y$,  $S$ is flat everywhere.
\end{proof}

\appendix

\section{Preliminaries on Initial Data Sets}\label{sec:af}

\subsection{Einstein constraint operators and dominant energy condition}
Let $n\ge 3$ and let $U$ be a connected $n$-dimensional smooth manifold, $g$ a Riemannian metric, and $k$ a symmetric $(0,2)$-tensor. We refer to such a triple $(U, g, k)$ as an \emph{initial data set}. We say that $(U, g, k)$ \emph{embeds} in a Lorentzian spacetime $(\mathbf{N}, \g)$ of one higher dimension if $(U, g)$ isometrically embeds into $(\mathbf{N}, \g)$ with $k$ as its second fundamental form. 

We denote by $G  = \Ric_{\g} - \tfrac{1}{2} R_{\g} \g$  the Einstein tensor of $\g$.  Standard computations give the following result, see, e.g. \cite[Appendix B]{Huang-Lee:2024}. 
\begin{lemma}
Let $(\N, \g)$ be a spacetime, and let $U$ be a spacelike hypersurface with the induced data $(g, k)$ and the future-directed unit normal $\n$. Let $e_i$ and $e_j$ be tangent vectors to $U$. Then components of the Einstein tensor take the form
\begin{align*}
	G(\n, \n)&= \tfrac{1}{2} (R_g - |k|^2 + (\tr_g k)^2)\\
	G(\n, e_i)&= (\Div_g k)_i - \nabla_i (\tr_g k)
\end{align*}
If $(\N, \g)$  admits a Killing vector field $\Y$ that is transverse to $U$,  and if $\Y = f\n + X$ along $U$. Then along $U$, the lapse-shift satisfies the relation:
\begin{align}\label{eq:LX}
	\mathscr L_X g = -2 k f.
\end{align}
Moreover, the spatial components of the Einstein tensor along $U$ are given by:
\begin{align}\label{eq:G}
\begin{split}
G(e_i, e_j)&= \left[R_{ij} -\tfrac{1}{2}R_g g_{ij}\right] \\
&\quad + \left[(\tr_g k)k_{ij} - 2k_{i\ell}k^\ell_j \right]
 + \left[ -\tfrac{1}{2}(\tr_g k)^2 +\tfrac{3}{2}|k|_g^2 \right]g_{ij}\\
&\quad+ f^{-1}\left[ -\tfrac{1}{2}(\mathscr L_X k)_{ij} + \tfrac{1}{2}\tr_g (\mathscr L_X k)g_{ij} -  f_{;ij} +(\Delta_g f) g_{ij} \right].
\end{split}
\end{align}

\end{lemma}

Since the components $G(\n, \n)$ and $G(\n, e_i)$ can be expressed entirely in terms of the initial data, they are well-defined even without assuming that the initial data set embeds  into a spacetime. It is conventional to call these the \emph{energy density} $\mu$ and the \emph{momentum density} $J$,  denoted by
\begin{align} \label{equation:mass-current}
\begin{split}
	\mu&:=\tfrac{1}{2}\left(R_g  - |k|^2 + (\tr_g k )^2 \right)\\
	J&:=  \Div_g k - d (\tr_g k).
\end{split}
\end{align}
We say that an initial data set $(U, g, k)$ satisfies the \emph{dominant energy condition} if $\mu \ge |J|_g$ holds everywhere on $U$; equivalently  $G(\n, \w)\ge 0$ for all future-directed causal vectors $\w$.  This condition is closely related to a stronger energy condition, referred to as the \emph{spacetime dominant energy condition}, which requires that $G(\u, \w)\ge 0$ for all future-directed causal vectors $\u, \w$ in the spacetime $(\N, \g)$.

The \emph{constraint map}, defined on initial data sets $(g, k)$, is given by 
\[
\Phi(g, k) = (\mu, J).
\]
 The linearization of $\Phi$ at $(g, k)$, denoted by $D\Phi|_{(g, k)}$, takes the form: for any pair of symmetric $(0,2)$-tensors $(h, w)$ on $U$,
\begin{align} \label{eq:phi}
\begin{split}
	&D\Phi|_{(g, k)} (h, w) \\
	&\quad = \Big( \tfrac{1}{2} L_g h+  k_i^p k_{\ell p} h^{i\ell}  -   k^{ij}w_{ij}  + (\tr_g k) (-k^{ij}h_{ij}  + \tr_g w), \\
	& \quad \quad \Div_g w - d(\tr_g w ) - (k_{ij,p} - k_{jp,i}) h^{jp} - g^{jp} k_i^q h_{pq;j} + \tfrac{1}{2} g^{jp} k_i^q h_{pj;q} + \tfrac{1}{2} k^{jq} h_{jq;i} \Big),
\end{split}
\end{align}
where $L_g h = -\Delta_g (\tr_g h )+ \Div_g \Div_g h - h\cdot \Ric $, and indices are raised using the metric $g$. 

Let $(f, X)$ be a pair of a scalar function and a vector field on $M$, which we will often call a \emph{lapse-shift pair}. The  $L^2$-formal adjoint operator $(D\Phi|_{(g, k)})^*$  is given as follows:
\begin{align*}
	&(D\Phi|_{(g,k)})^* (f,X)\\
	&=\left(  \tfrac{1}{2} (-\Delta f g_{ij} + f_{;ij} - fR_{ij} ) + (k^{ip} k_{jp} -( \tr k )k_{ij})f\right. \\
	&\quad\left. + \tfrac{1}{2}( k_{ji;\ell}X^\ell +   k_{\ell j}X^\ell_{;i} +   k_{\ell i}X^\ell_{;j} - g_{ji} k_{q\ell, q} X^\ell - g_{ij} k^{q\ell} X_{\ell; q} - k_{ij} \mathrm{div}X),  \right.\\
	&\quad \quad  \left.-\tfrac{1}{2}(X_{i;j}+X_{j;i}) +(- k_{ij} + g_{ij} (\tr k))f+(\mathrm{div} X) g_{i j} \right).
\end{align*}

\begin{definition}[{\cite[Definition 3.1]{Huang-Lee:2024}, \cite[Section 2.2]{Corvino-Huang:2020}}]
Let $\varphi$ be a scalar function on $U$. The $\varphi$-\emph{modified constraint operator} with respect to $(g, k)$ is defined on initial data $(\gamma, \tau)$ on $U$ by 
\begin{align}\label{eq:varphi}
    \Phi^\varphi_{(g, k)} (\gamma, \tau) =\Phi(\gamma, \tau) + \big(0,\tfrac{1}{2} \gamma \cdot J\big) + \left(\varphi |J|_\gamma^2 \, , \varphi |J|_g \gamma \cdot J\right)
\end{align}
where $J=\Div_g k - d(\tr_g k)$ is the momentum density associated with the fixed background $(g, k)$,  and $(\gamma \cdot J)_i = g^{j\ell} \gamma_{ij} J_\ell$. Among this infinite-dimensional family of modified constraint operators, the case $\varphi =0$ is of particular importance, which is denoted by
\begin{align} \label{eq:varphi0}
\overline{\Phi}_{(g,k)} : = \Phi^{0}_{(g,k)}.
\end{align}

\end{definition}

We denote the linearization  of $\Phi^\varphi_{(g, k)}$ at the background initial data set $(g, k)$ by $D \Phi^\varphi_{(g, k)} :=D \Phi^\varphi_{(g, k)} |_{(g,k)} $ and the corresponding $L^2$-formal adjoint operator by $(D\Phi^\varphi_{(g, k)})^* $. Then 
\begin{align}
	D \Phi^\varphi_{(g, k)} (h, w)&= D\Phi|_{(g, k)} (h, w) +\big (0, \tfrac{1}{2}  h \cdot J\big) \label{eq:phi2} \\
	&\quad + (\varphi h_{ij} J^i J^j, \varphi |J|_g h\cdot J) \notag\\
    (D \Phi^\varphi_{(g, k)})^* (f, X)& =\left( D\Phi|_{(g,k)} \right)^*(f, X) + \big( \tfrac{1}{2} X\odot J, 0\big) \notag \\
    &\quad +\big (\varphi J \odot (fJ+ |J|_g X), 0\big). \notag
\end{align}

\begin{remark}
Note that the definitions of the modifiedconstraint operators given here closely follow \cite{Huang-Lee:2020, Huang-Lee:2024}, except for the following differences:
\begin{enumerate}
\item We use the second fundamental form $k$  instead of the conjugate momentum tensor $\pi^{ij}  := g^{i k}  g^{j\ell} k_{k\ell} -( \tr_g k)g^{ij}$, as used in  \cite{Huang-Lee:2020, Huang-Lee:2024}.
\item We define the constraint operator as $\Phi(g, k) = (\mu, J)$, rather than $(2\mu, J)$ as in \cite{Huang-Lee:2020, Huang-Lee:2024}. As a result, all formulas involving the lapse-shift pair $(f, X)$ here correspond to those in \cite{Huang-Lee:2020, Huang-Lee:2024} upon replacing $f$ there with $\tfrac{1}{2} f$ (or, in some cases, replacing $X$ with $2X$). This normalization is more natural geometrically, since $f\n + X$ directly represents the spacetime vector field, in contrast to 
 $2f\n +X$ in the earlier works.
 \item In  \cite[Definition 3.1]{Huang-Lee:2024},  more general $(\varphi, Z)$-modified constraint operators are considered. However, in this paper it suffices to take $Z =J$, and we therefore omit the superscript $Z$. 
\end{enumerate}
\end{remark}

The main advantage of   the $\varphi$-modified constraint operator $\Phi^\varphi_{(g,k)}$ is that controlling the modified constraints of $(\gamma, \tau)$ near $(g,k)$ yields good control over the dominant energy condition, as specified in the following lemma. It is a straightforward generalization of the case $W=0$ from \cite[Lemma 3.4]{Huang-Lee:2024}.

\begin{lemma}\label{lemma:sigma_preserve}
Let $(g, k)$ and $(\gamma, \tau)$ be initial data on a manifold $U$ with $|\gamma-g|_g<1$, and let $\varphi$ be a function on $U$ such that $|\varphi J|_g < \frac{\sqrt{2}-1}{2}$, where $(\mu, J)$ is the energy and momentum densities of $(g, k)$. Suppose
\begin{equation}\label{equation:prescribe}
\Phi^{\varphi}_{(g, k)} (\gamma, \tau  ) = \Phi^{\varphi}_{(g, k)} (g, k ) + (u, W)
\end{equation}
for some function $u$ and one-form $W$, with either $J$ or $W$ vanishing at each point.  Then the energy and momentum densities $(\bar \mu, \bar J)$ of $(\gamma, \tau)$ satisfies
\[
	\bar \mu - |\bar J|_{ \gamma} \ge \mu - | J|_{g}  + u - |W|_g -  \tfrac{1}{2} |\gamma- g|_g |W|_g.
\]	
In particular, if $(g, k)$  satisfies the dominant energy condition and $\Phi^{\varphi}_{(g, k)} (\gamma, \tau  ) = \Phi^{\varphi}_{(g, k)} (g, k )$, 
 then $(\gamma, \tau)$ also satisfies the dominant energy condition. 
\end{lemma}
\begin{proof}
Let $h:=\gamma -g$. In what follows,  all lengths and inner products are computed using the metric $g$, unless otherwise specified.
  
Rewriting the hypothesis~\eqref{equation:prescribe} using the definition of the $\varphi$-modified constraint operator~\eqref{eq:varphi}, we obtain
\[
	\Phi(\gamma , \tau) = \Phi(g, k) + \left(-\varphi \langle h\cdot J,  J\rangle, - \left(\tfrac{1}{2} + \varphi |J|\right) h\cdot J\right) + (u, W),
\]
or equivalently,
\begin{align*}
	\bar{\mu} &= \mu  -\varphi  \langle h\cdot J,  J\rangle + u\\
	\bar{J} &=J - \left(\tfrac{1}{2} + \varphi |J| \right) h\cdot J + W.
\end{align*}
Following the same computations as in \cite[Lemma 3.4]{Huang-Lee:2024}, we have
\begin{align*}
	|\bar J|^2_{\gamma}& = (g+h)_{ij} \bar J^i \bar J^j  \\
	&\le (|J| - \varphi  \langle h\cdot J,  J\rangle)^2 + (1+|h|)|W|^2,
\end{align*}
where we used the assumption that  either $J$ or $W$ vanishes at each point, so there are no cross terms involving contractions between $J$ and $W$. Combining the square root of the above inequality with the  expression for $\bar \mu$ gives the desired inequality. 
\end{proof}

\subsection{Asymptotically flat initial data sets}

Let $n\ge 3$ and $M$ be a connected, $n$-dimensional manifold, possibly with compact boundary. Fix parameters $q\in\left(\tfrac{n-2}{2}, n-2\right)$,  $\alpha\in(0,1)$, and $\epsilon>0$. We say that an initial data set $(M, g, k)$ is \emph{asymptotically flat} (of type $(q, \alpha, \epsilon)$) if  there exist compact subsets $K_0 \subset M$ and $B\setminus \mathbb R^n$ and a diffeomorphism $M\setminus K  \cong \mathbb{R}^n \setminus B$ such that 
\[
	(g - g_{\mathbb{E}}, k) \in \C^{2,\alpha}_{-q} (M)\times \C^{1,\alpha}_{-1-q} (M)
\]
and, for some $\epsilon>0$,  
\begin{align} \label{eq:muJ}
	(\mu, J) \in \C^{0,\alpha}_{-n-\epsilon} (M),
\end{align}
%\margin{Do you think that the Lagrange multiplier method could work for $\mu,J\in L^1$ by choosing different function spaces?}
where $g_{\mathbb{E}}$ is a Riemannian background metric on $M$ that coincides with the Euclidean metric in the asymptotic coordinate chart $M\setminus  K_0  \cong \mathbb{R}^n \setminus  B$. %When the decay condition \eqref{eq:muJ} is not imposed, we simply say that $(M, g, k)$ is asymptotically flat of type $(q, \alpha)$.  
Note that under this definition, the manifold $(M, g)$ is necessarily complete, possibly with boundary.  We denote by $\mathcal M^{2,\alpha}_{-q}(M)$ the space of Riemannian metrics $g$ on $M$ such that $g - g_{\mathbb{E}}\in \C^{2,\alpha}_{-q} (M)$. 

The spaces $\C^{k,\alpha}_{-q}$ referred to above are \emph{weighted H\"older spaces} (as defined in~\cite{Huang-Lee:2020}, for example), defined with respect to a fixed asymptotically flat coordinate chart $\{ x^1, \dots, x^n \}$ arising from the chosen diffeomorphism. Let $|x| = \sqrt{(x^1)^2+\cdots + (x^n)^2}$.  Our convention is  that $f\in \C^{k,\alpha}_{-q}(M)$ if and only if $f\in \C^{k,\alpha}_{\mathrm{loc}}(M)$ and there is a positive constant $C$ such that, for any multi-indices $I$ with $|I|\le k$,
\begin{align*}
	|(\partial^I f)(x)|\le C|x|^{-|I|-q} \quad \mbox{ and } \quad [f]_{k,\alpha; B_1(x)} \le C|x|^{-k-\alpha-q}
\end{align*}
on $M\setminus  K_0$.  We use the notation $O^{k,\alpha}(|x|^{-q})$ to denote an arbitrary function or tensor components in $\C^{k,\alpha}_{-q}$, and we write $O(|x|^{-q})$ for $O^0(|x|^{-q})$.

The \emph{ADM energy} $E$ and the \emph{ADM linear momentum} $P=(P_1, \dots, P_n)$ of an asymptotically flat initial data set $(M, g, k)$ are defined by
\begin{align*}
	E&:= \tfrac{1}{2(n-1)\omega_{n-1}} \lim_{k\to \infty} \int_{|x|=r}\sum_{i,j=1}^n (g_{ij,i}-g_{ii,j})\nu^j \, d\sigma_0\\
	P_i &:= \tfrac{1}{(n-1)\omega_{n-1}} \lim_{k\to \infty} \int_{|x|=r} \sum_{i,j=1}^n (k_{ij}  - (\tr_g k)g_{ij} )\nu^j \, d\sigma_0.
\end{align*}	
These integrals are taken over the coordinate spheres  $\{ |x|=r \}$ in the asymptotic region  $M\setminus  K \cong \mathbb{R}^n \setminus  B$. Here, $\nu^j$ denotes the Euclidean outward unit normal, $d\sigma_0$ is the measure induced by the Euclidean metric,  $\omega_{n-1}$ is the volume of the standard $(n-1)$-dimensional unit sphere, and the commas denote partial derivatives with respect to the coordinates. The condition $q>\frac{n-2}{2}$, together with the asymptotics \eqref{eq:muJ},  ensures that the limits defining the ADM energy-momentum exist. When  $E\ge |P|$, the \emph{ADM  mass} is defined as $m_{\mathrm{ADM}} := \sqrt{ E^2 -|P|^2}$. We say the \emph{ADM mass is zero} if $E=|P|$. 

The next two results describe the asymptotic behavior of a lapse-shift pair on asymptotically flat initial data sets.

\begin{lemma}[{Cf. \cite{Beig-Chrusciel:1996, Chrusciel-Maerten:2006}}]\label{le:asy}
Let $(M, g, k)$ be an asymptotically flat initial data set,  and let $(f, X)\in \C^{2,\alpha}_{\mathrm{loc}}$ be a lapse-shift pair. In the asymptotically flat exterior,  write the components as $Z=(f, X^1, \dots, X^n)$. Suppose that, with respect to the asymptotically flat coordinate chart, $Z$ satisfies the linear system of Hessian equations:
 \begin{align}\label{eq:Hessian}
	Z_{i,jk} + A_{ijk pq}Z_{p, q} + B_{ijk\ell} Z_\ell = O^{0,\alpha}(|x|^{-2-q}),
\end{align}
where the coefficient matrices satisfy $ A_{ijk pq} =O^{1,\alpha}(|x|^{-1-q})$ and $B_{ijk\ell} = O^{0,\alpha} (|x|^{-2-q})$. 

Then $(f, X)$ has precisely one of the following asymptotics:
\begin{enumerate}
\item $f = a_i x^i + O^{2,\alpha}(|x|^{1-q})$ and $X^j= b_{j\ell} x^\ell + O^2(|x|^{1-q})$ for  constants $a_i, b_{j\ell}$, not all zero. 
\item  $(f, X) = (a, b) + O^{2,\alpha}(|x|^{-q})$ for contants $(a, b)\in \mathbb R\times \mathbb R^n$, not all zero.\label{it:2}
\item  $(f, X) = O^{2,\alpha}(|x|^{-q})$. 
\end{enumerate}
Furthermore, suppose
\begin{align}\label{eq:Hessian0}
Z_{i,jk} + A_{ijk pq}Z_{p, q} + B_{ijk\ell} Z_\ell =0.
\end{align}  
If  $(f, X)\to 0$, then $(f, X)\equiv 0$.
\end{lemma}
\begin{remark}
The condition \eqref{eq:Hessian} or \eqref{eq:Hessian0} holds in the following cases used in this paper.
\begin{enumerate}
\item If $(f, X)$ solves $(D \Phi^\varphi_{(g, k)})^* (f, X)\in \C^{0,\alpha}_{-2-q} (M)\times \C^{1,\alpha}_{-1-q}(M)$ for any bounded smooth function $\varphi$, then $(f, X)$ satisfies \eqref{eq:Hessian}. 
\item Let $(\N, \g)$ be a spacetime admitting a Killing vector field $\Y$ and an asymptotically flat initial data set $(M, g, k)$. If $G(\partial_i, \partial_j) = O^{0,\alpha}(r^{-2-q})$ for asymptotically flat coordinate vectors of $M$, then the lapse–shift pair $(f, X)$ of $\Y$ along $M$ satisfies \eqref{eq:Hessian0} by \eqref{eq:LX} and \eqref{eq:G}.
\end{enumerate}
\end{remark}
\begin{proof}
We overview the arguments in \cite[Appendix C]{Beig-Chrusciel:1996}. The Hessian equations imply that  $Z$ satisfies a second-order ODE along each fixed angular direction, for $r := |x| \in [r_0, \infty)$, 
\[
	Z''(r) = A(r) Z'(r) + r B(r) Z(r) + O(r^{-2-q})
\]
where the matrix coefficients satisfy $|A(r)|+ r|B(r)| = O(r^{-1-q})$. An ODE and bootstrap argument then shows that $Z$ grows at most linearly: $|Z(r)|\le  C r$ for some constant $C$.  Taking the trace of the Hessian equations for $(f, X)$ shows that $(f, X)$ also satisfies a system of Laplace equations. Since $(f, X)$ has at most linear growth, the asymptotic result follows from harmonic expansion. If \eqref{eq:Hessian0} holds, the bootstrap argument implies that $Z$ vanishes to infinite order at infinity, and thus must be identically zero.
\end{proof}

We say that  the lapse-shift pair $(f, X)$ is \emph{asymptotically translational} if there exists $(a, b)\in \mathbb R\times \mathbb R^n$ such that $(f, X)- (a, b) \in \C^{2,\alpha}_{-q}(M)$. It follows from \cite[Proposition 3.1 and Theorem 3.4]{Beig-Chrusciel:1996} and \cite[Theorem A.6]{Huang-Lee:2020} that $(f, X)$ admits the following more detailed asymptotic expansions. (Note that the coefficients differ slightly from those in \cite[Theorem A.6]{Huang-Lee:2020} due to a different normalization, replacing $X$ and $b$ there with $2X$ and $2b$, respectively.)

\begin{theorem}\label{th:expan}
Let $(M, g, k)$ be asymptotically flat.  Let $(f, X)$ be asymptotically translational  to $(a, b)$. Suppose $(f, X)$ satisfies, for some $\epsilon>0$, 
\begin{align}\label{eq:Rij}
\begin{split}
	-\Delta f g_{ij} + f_{;ij} - f R_{ij} + \tfrac{1}{2} k_{ij;\ell} X^\ell &= O^{0,\alpha}(|x|^{-n-\epsilon})\\
	X_{i;j} + X_{j;i} + 2 k_{ij} f &=O^{0,\alpha}(|x|^{1-n-\epsilon}).
\end{split}
\end{align}
Then, in asymptotically flat harmonic coordinates $\{ x \}$, we have the following expansions,  for some constant $q_1>0$, constant coefficients $c_{ij}$, and a function $\phi\in \C^{3,\alpha}_{1-q}$ solving $\Delta_{g_{\mathbb E}} \phi = (1-n) \tr_{g_{\mathbb E}}  k$:
\begin{align*}
	f &= a + \left ( -aE + \tfrac{1}{n-2} b\cdot P\right) |x|^{2-n} + \tfrac{1}{n-1} b_k \phi_{,k}+ O^{2,\alpha}(|x|^{2-n-q_1})\\
	X^i &=b_i  + \big(-  \tfrac{2(n-1)}{n-2} E b_i + b_k  c_{ik}\big) |x|^{2-n} +\tfrac{1}{n-1} a \phi_{,i}+ O^{2,\alpha}(|x|^{2-n-q_1}).
\end{align*}
\end{theorem}
\begin{remark}
Comparing with the formulas in \cite[Theorem A.6]{Huang-Lee:2020} involving the vector field $V$, we note that the terms here involving the coefficients $c_{ij}$ arise from $V$. Using that $\{x\}$ is a harmonic coordinate chart, we have $\tfrac{1}{2} g_{jj,i} - g_{ij,j} = O^{1,\alpha}(r^{-1-2q})$, so that $\Delta_0 V_i = O^{1,\alpha}(r^{-1-2q})$ as in \cite[Eq.~(A.7)]{Huang-Lee:2020}. This yields the harmonic expansion $V_{i,j} = c_{ij} |x|^{2-n} + O(r^{-2q})$.
\end{remark}
\begin{remark}\label{re:Rij}
Let $(\N, \g)$ be a spacetime admitting a Killing vector field $\Y$ and contain an asymptotically flat initial data set $(M, g, k)$ . Let $(f, X)$ be the lapse-shift pair of $\Y$ along $M$. If the Einstein tensor $G(\partial_i, \partial_j) = O(r^{-n-\epsilon})$ for coordinate vectors $\partial_i, \partial_j$ to $M$, then $(f, X)$ satisfies \eqref{eq:Rij}  by \eqref{eq:LX} and \eqref{eq:G}.
\end{remark}

\section{Degenerate elliptic equations} \label{se:de}

We review the following classical results for degenerate equations.  Let $(U, g)$ be a Riemannian manifold and $\Omega\subset U$ be a compact domain with smooth boundary. Let $f>0$ and $\hat X$ be $\C^1$ on $U$. Let 
\begin{align*}
	Lw &:= \Div_g \big( f (\nabla w - (\nabla w \cdot \hat X)\hat X) \big) + b\cdot \nabla w
\end{align*}
where $f>0$ and $b$ is continuous. Denote the principal coefficients $\g^{ij} = g^{ij} - \hat X^i \hat X^j$.

We say that $p\in \partial \Omega$ is a \emph{characteristic point} for $(L, \Omega)$ if the unit normal $\nu$ to $\partial \Omega$ satisfies $\g^{ij} \nu_i \nu_j=0$ at $p$; equivalently, $\nu$ is proportional to $\hat X$ if $|\hat X|=1$. We overview the  proofs from \cite{Hill:1970} in the next two results, adapted to our notation for the reader’s convenience.

\begin{lemma}[{\cite[Lemma 3]{Hill:1970}}]\label{le:hopf}
  Suppose $Lw \le 0$ and $w\ge 0$ on $U$. Let $\Omega\subset \Int U$ be a compact subset with smooth boundary. If $w>0$ on $\Int \Omega$, and $w=0$ and $|\hat X|=1$ at $p\in \partial \Omega$, then $p$ is a characteristic point for $(L, \Omega)$ and $Lw(p)=0$.  
\end{lemma}
\begin{proof}
Since $w$ attains its minimum  at $p$ on $U$, we have $\nabla w=0$ and $\nabla^2 w \ge 0$ at $p$, which implies $Lw(x)=0$. Thus it remains to show that $p$ is characteristic, which follows by adapting the Hopf boundary point lemma to the degenerate setting.

Choose $x_{0}\in \Int \Omega$ close to $p$, and let $d$ be the distance to $x_{0}$ with $R=d(p)$. Denote by $B_{R}=\{x: d(x)<R\}\subset \Int \Omega$ the ball centered at $x_0$ of radius $R$.  Then the normal vectors of $\partial B_R $ and $\partial \Omega$ coincide at $p$. Suppose, to the contrary, that $p$ is not characteristic. Then the outward unit normal $\nu=\nabla d$ to $\partial \Omega$ at $p$ is not proportional to $\hat X$.

For constant $\lambda>0$, define 
\[
	h(x) = e^{-\lambda R^2} - e^{-\lambda d^2(x)} \le 0 \quad \mbox{ on } B_R. 
\]
A direct computation shows that 
\begin{align*}
	Lh %&=f \g^{ij}  h_{,i} h_{,j} + O(\lambda) = - 4 f\lambda^2 d^2 e^{-\lambda d^2} \g^{ij} d_{,i} d_{,j} + O(\lambda) \\
	&= - 4 f\lambda^2 d^2 e^{-\lambda d^2}  (|\nabla d|^2 - (\nabla d \cdot \hat X)^2) + O(\lambda),
\end{align*}
so for large $\lambda$, we have $Lh<0$ on $B(p)\cap B_R$ for some neighborhood $B(p)$ of $p$. 

For sufficiently small $\epsilon>0$, we have $L (w +\epsilon h)<0$ on $B(p)\cap B_R$, while $w+\epsilon h \ge 0$ on $B(p)\cap B_R$ and  vanishes at $p$. The strictly maximum principle then gives $w+\epsilon h >0 $ inside $B(p)\cap B_R$. It implies that $\nu(w) \neq 0$ at $p$, which contradicts that $w$ attains a minimum at $p\in \Int U$. 
\end{proof}

\begin{lemma}[{\cite[Theorem 1]{Hill:1970}}]\label{le:pro}
Suppose  $Lw\le 0 $ and $w\ge 0$ on $U$. We define the null set $Z = \{ p\in U: w(p)=0\}$, and the \emph{propagation set} of  $z\in Z$  in $ U$ by
\begin{align*}
	P_{z}= &\{ x\in \Int U:  \mbox{ $x, z\in \gamma(t)$ where $\gamma(t)$ is an integral curve of $V\in \hat X^\perp$}\}.
\end{align*}
Assume that $|\hat X|=1$ on $Z$. Then $P_{z} \subset Z$. 
\end{lemma}
\begin{proof}

Our goal is to show that $w(x)=0$ for all $x\in P_{z}$. Suppose, on the contrary, that there exists $q\in P_z$ with $w(q)>0$. Let $\gamma(t)$ be an integral curve of $V\in \hat X^\perp$ from $z = \gamma(0)$ to $q=\gamma(t)$ for some $t>0$. Let $0\le t_0<t$ be the first $t_0$ such that $w( \gamma(t_0))=0$ when going back from $t$. That is, $t_0 = \max \{ 0 \le s < t: w(\gamma(s))=0\}$ (possibly $t_0=0$.)

We will derive a contradiction that $t_0$ is not the maximum. We shift the parameter  $\zeta(t) = \gamma(t-t_0)$ so that $\zeta(0) = \gamma(t_0)$. We will work on a small neighborhood of $\zeta(0)$, so we may assume there exists a coordinate chat $\{ x^1, \dots, x^n\}$ around $\zeta(0)$ such that $x^1$ is the flow line of $V$, and we compute in the Euclidean distance function below, just as in \cite{Hill:1970}.

Define, for $t\ge 0$, 
\[
 	\rho(t) = \inf_{y\in Z} | \zeta(t) - y| \ge 0, 
\]
 the distance function from $\zeta(t) $ to  $Z$. Then $\rho$ is Lipschitz (hence absolutely continuous) with $\rho(0)=0$ and $\rho(t)>0$ for small $t>0$. Let $B(t)$ be the ball centered at $\zeta(t)$ of radius $\rho(t)$. Since $w>0$ on $B(t)$ with $w=0$ for some point $\hat z(t)$ on $\partial B(t)$. Lemma~\ref{le:hopf} implies that $\hat X$ is normal to $\partial B(t)$ at $\hat z(t)$. Thus, for any $\epsilon>0$ and all small $t$, $\hat z(t)\in \partial B(t) \cap \{ (x^1, \dots, x^n) : |x^1 - \zeta^1 (t)| < \epsilon\}$, the  portion of the sphere possibly orthogonal to $\hat X$. Here $\zeta^1$ denotes the first component of $\zeta$ in  $\{x \}$ coordinates. 
 
 From there, one can then estimate that $\rho(t+\delta) - \rho(t) \le C \delta \rho(t)$ for some constant $C>0$ via elementary geometry.  It then implies that  $\rho'(t) \le C \rho(t)$ for a.e. $t$. Together with the initial condition $\rho(0)=0$ and $\rho(t)\ge 0$, we conclude that $\rho(t) \equiv 0$ for all $t>0$ small. This leads to a contradiction. 
\end{proof}

\section{Local Surjectivity}\label{se:sur}
In this section, we prove Proposition~\ref{pr:sur}, restated below for the reader’s convenience.

\begin{manualproposition}{\ref{pr:sur}}
Let $(M, g, k)$ be an asymptotically flat initial data set of type $(\alpha, q, \epsilon)$ and let $\varphi$ be a smooth bounded scalar function on $M$.  Let $m\ge 2$. 
\begin{enumerate}
    \item \label{it:nobd} Suppose $M$ is complete without boundary. Then the smooth map 
    \[
     \Phi^\varphi_{(g, k)}:\mathcal{M}^{m,\alpha}_{-q}(M)\times \C^{m-1,\alpha}_{-1-q}(M)\to  \C^{m-2, \alpha}_{-2-q}(M)
    \] is locally surjective.
    \item \label{it:bd} Suppose $M$ is complete with  boundary~$\partial M$. Then  the smooth map 
    \begin{align*}
    T:\mathcal M^{m,\alpha}_{-q} (M)&\times \C^{m-2,\alpha}_{-1-q}(M)\to \C^{m-2, \alpha}_{-2-q}(M) \times \mathcal B^{m,\alpha}(\partial M)\\
    	&T(g, k) = \Big(  \Phi^\varphi_{(g, k)}, B(g, k)\Big)
    \end{align*}
is locally surjective at $(g, k)$. 
\end{enumerate}
\end{manualproposition}

The first case, where $M$ has no boundary, is standard; see, for example, \cite[Proposition 3.11]{Corvino-Schoen:2006} and \cite[Lemma 2.10]{Huang-Lee:2020}. While those results are stated for weighted Sobolev spaces, the arguments readily extend to weighted H\"older spaces. (The additional terms involving  $\varphi$ are of zeroth order and thus do not affect the argument.) The second case for the classical Einstein constraint operator, where the boundary is nonempty, was proven by Z.~An~\cite[Proposition 2.2]{An:2020} for $n=3$. Here, we provide an alternative proof that extends the result to general $n \geq 3$, building on the arguments for the time-symmetric case in \cite[Appendix A]{Huang-Jang:2022}.

Recall the formula $D\Phi^\varphi_{(g, k)} (h,w)$ from \eqref{eq:phi} and \eqref{eq:phi2}:
\begin{align*}
D\Phi^\varphi_{(g, k)} (h,w) = \Big(\tfrac{1}{2} &(-\Delta_g (\tr_g h )+ \Div_g \Div_g h) + F_1(h,w), \\
&\quad \Div_g w - d( \tr_g w)+F_2(h,w)\Big) 
\end{align*}
where $F_1(h,w)$ denotes the terms that are zeroth order in $h, w$, and $F_2(h,w)$ denotes the  terms  involving at most one derivative of  $h$ and zeroth order in~$w$. 

The linearized mean curvature and unit normal vector are given respectively, for $a, b=1,\dots, n-1$, by
\begin{align*}
	\nu '|_g(h)&=-\tfrac{1}{2} h (\nu_g, \nu_g) \nu_g - g^{ab} h(\nu_g, e_a) e_b \\
	H'|_g(h)&= \tfrac{1}{2} \nu_g(\tr_g h^\intercal) - \Div_{g^\intercal} h(\nu_g)^\intercal - \tfrac{1}{2} h(\nu_g, \nu_g)H_g
 \end{align*}         
 where $e_1, \dots, e_{n-1}$ are tangent to $\partial M$.  The Bartnik boundary data map is given by $B(g, k) = (g^\intercal,  H_g, k(\nu_g)^\intercal,  \tr_{g^\intercal} k)$, and the linearization:
\begin{align*}
	&B'|_{(g, k)}(h,w) \\
	&=\Big ( h^\intercal , H'|_g(h), w(\nu_g)^\intercal + k (\nu'|_g(h))^\intercal,  \tr_{g^{\intercal}} w - h^\intercal \cdot k^\intercal \Big)\\
	&=\Big ( h^\intercal , \tfrac{1}{2} \nu_g(\tr_g h^\intercal) - \Div_{g^\intercal} h(\nu_g)^\intercal   + F_3(h), w (\nu_g)^\intercal + F_4(h),  \tr_{g^\intercal} w + F_5(h) \Big)
\end{align*}
where  $F_3 (h), F_4(h)$ denotes the terms involving only zeroth order in $h$. 

Define 
\begin{align*}
	\mathcal S^{m,\alpha} = \Big\{ (h, w) \in  \C^{m,\alpha}_{-q}(M)\times \C^{m-1,\alpha}_{-1-q}(M): B'|_{(g, k)}(h,w)=0\Big\}.
\end{align*}
To prove Item~\eqref{it:bd} Proposition~\ref{pr:sur}, it suffices to establish the following lemma.
\begin{lemma}
Let $(M, g, k)$ be a complete, asymptotically flat initial data set  with nonempty boundary $\partial M$. Then the map $D\Phi^\varphi_{(g, k)} : \mathcal S^{m,\alpha}  \to\C^{m-2,\alpha}_{-2-q}(M)$ has finite-dimensional cokernel, and hence it has closed range. 
\end{lemma}
\begin{proof}
For a scalar function $v$, a  vector field $X$, and a pair of symmetric (0,2)-tensors $(h, w)$, we define $\hat L_X g = L_X g - \tfrac{1}{n-1} (\Div_g X ) g$ and the differential operator $L$   on $M$ by
\begin{align*}
	&L(v, X, h, w) \\
	&=\left\{ \begin{array}{l} \frac{1}{2} (1-n) \Delta v + \tfrac{1}{2} (-\Delta_g (\tr_g h )+ \Div_g \Div_g h)  + F_1 (vg+h, \hat L_X g+w)\\
	 \Delta_g X+ \Div_g w - d(\tr_g w) + F_2(vg+h, \hat L_X g + w)  \\
	 \Delta h\\
	 \Delta w \end{array} \right..
\end{align*}
The first two equations in $L(v, X, h, w)$ arise from $D \Phi^\varphi_{(g, k)}(vg+h, \hat L_X g + w)$, while the last two serve as ``gauge'' equations to ensure that $L$ forms an elliptic system (see \eqref{eq:block} below). 

We define the boundary operator on $\partial M$ by 
\begin{align*}
	 &P(v, X, h, w) \\
	 &= \left\{ \begin{array}{l} vg^\intercal + h^\intercal \\
	  \tfrac{1}{2} (n-1)\nu(v) + \tfrac{1}{2} \nu(\tr h^\intercal) - \Div_{g^\intercal} h(\nu)^\intercal + F_3(vg+h)\\
	 L_X g  (\nu_g)^\intercal + (w(\nu_g))^\intercal  + F_4(vg+h)\\
\tr_{g^\intercal} L_X g - \Div_g X + \tr_{g^\intercal} w + F_5 (vg+h)\\
		 h(\nu, \nu)\\
	 (\Div_g h - dv)^\intercal \\
	w \end{array} \right..
\end{align*}
The first four equations in $P(v, X, h, w)$ arise from $B(vg+h, \hat L_X g + w)$, while the rest of equations in $P(v, X, h, w)$  are ``gauge'' equations,  chosen to form a complementing boundary condition, specifically, so that \eqref{eq:bd0} holds. 

We verify that $L(v, X, h, w)$ is elliptic and that the boundary operator $P(v, X, h, w)$ satisfies the complementing boundary condition. Once these properties are established, $L$ is a Fredholm operator on $\mathcal S^{m,\alpha}$, so $\Range L$ has finite codimension. Since $\Range D \Phi^\varphi_{(g, k)}$ contains the range of the first two equations of $L$, it follows that $\Range D  \Phi^\varphi_{(g, k)}$ also has finite codimension.

%Let $\{ e_1, \dots, e_n\}$ be an orthonormal local frame so that $e_n = \nu$ along $\Sigma$. 
Let $N = 1 + n + n(n+1)$ denote the total number of components of $v, X, h, w$ with respect to a local frame. Let $u = (u_1, u_2, \dots, u_N)$ be a vector-valued function representing the components $(v, X_1,\dots, X_n,  h_{ij} , w_{ij})$. (The specific ordering of $h_{ij}$, $w_{ij}$ is not important.) We express the $i$-th (scalar) differential equation of $Lu$ and the $a$-th (scalar) differential equation of $Pu$ respectively by, for $i, a=1,\dots, N$, 
\begin{align*}
	(Lu)_i &= \sum_{j=1}^N L_{ij}(\partial) u_j\\
	(Pu)_a&=\sum_{j=1}^N P_{aj}(\partial ) u_j.
\end{align*}
To identify the principal symbols, we substitute the differential operator $\partial$ in $L_{ij}(\partial)$ and $P_{aj} (\partial)$ with polynomials in $\xi = (\xi_1, \dots, \xi_n)$. Let $L'_{ij}(\xi)$  denote the homogeneous degree-2 terms in $L_{ij}(\xi)$. The principal symbol of the differential operator $L$, denoted by $L'(\xi)$, is the $N\times N$ matrix whose $(i, j)$-th entry is $L'_{ij}(\xi)$. We write the principal symbol $L'(\xi)$ in the following $[1]+[N-1]$ block form:
 \begin{align} \label{eq:block}
\begin{split}
	L'(\xi)&=  \left[\begin{array}{c|c} \tfrac{1}{2}(1-n)|\xi|^2& * \\ 
	\hline 0 & |\xi|^2 I_{(N-1)}  
	\end{array}\right]\\
	\Det L'(\xi) &= \tfrac{1}{2} (1-n) |\xi|^{2N} \neq 0\mbox{ for all nonzero $\xi \in \mathbb R$}.
\end{split}
\end{align}
Thus,  $L$ is an elliptic operator in the sense of  \cite[p. 39]{Agmon:1964us}. %(Implicitly, we use the weights $s_i=0$, $t_j=2$ for all $i,j$ where the weights $s_i, t_j$ (as well the $r_h$ that appears below) are defined as in \cite{Agmon:1964us}.)

We now identify the principal symbol of the boundary operator $P$. Let $P'_{aj}(\xi)$ denote the terms in $P_{aj}(\xi)$ corresponding to the highest-order derivatives in each equation $(Pu)_a$.  The principle symbol, denoted $P'(\xi)$, is an $N\times N$ matrix whose $(a, j)$-th entry is $P'_{aj}(\xi)$. Writing $\xi= (\xi', \xi_n)$ and $\xi'$ is tangential and $\xi_n$ is normal to $\partial M$, we aim to show that, upon setting $\xi_n = \mathsf i$,  
\begin{align} \label{eq:bd0}
	\Det P'(\xi', \mathsf i) \neq 0
\end{align}
which is equivalent to verifying that the equation $P'(\xi', \mathsf i)\begin{bmatrix} v\\ X \\ h\\ w\end{bmatrix} =0$ has only the trivial solution $(v, X, h, w)=0$. That is, we must verify that the following linear system admits only the zero solution. Below, we let $\alpha, \beta= 1, \dots, n-1$ be tangential and $i, j= 1,\dots, n$: 
\begin{align*}
	&v\delta_{\alpha \beta} + h_{\alpha\beta}=0\\
	&\tfrac{1}{2} (n-1) \mathsf{i} v + \tfrac{1}{2}  \mathsf{i}\sum_{\alpha=1}^{n-1} h_{\alpha\alpha} -\sum_{\alpha=1}^{n-1} \xi_\alpha h_{n\alpha}=0 \\
		& \xi_\alpha X_n+ \mathsf{i} X_\alpha  = 0\\
	& \sum_{\alpha=1}^{n-1} \xi_\alpha X_{\alpha} - \mathsf{i} X_n=0 \\
	&\left(\sum_{\beta=1}^{n-1} \xi_\beta h_{\beta \alpha} \right) + \mathsf{i} \, h_{n\alpha} - \xi_\alpha v=0\\
		& h_{n n}=0\\
	&\omega_{ij} =0.
\end{align*}
This verification is straightforward.

%Rather than writing $P'(\xi)$ in the matrix form, it is more transparent to list all the rows of the matrix product $B'(\xi) \begin{bmatrix} u_1\\ \vdots\\ u_N\end{bmatrix}$ (and substitute $(u_1, \dots, u_N)$ with $(v, X, h, w)$). For $\alpha, \beta = 1,\dots, n-1$ and $i, j=1,\dots, n$:
%\begin{align*}
%	&v\delta_{\alpha \beta} + h_{\alpha\beta}\\
%	&\tfrac{1}{2}(n-1) \xi_n v + \tfrac{1}{2} \xi_n \sum_{\alpha=1}^{n-1} h_{\alpha\alpha} -\sum_{\alpha=1}^{n-1} \xi_\alpha h_{n\alpha} \\
%		&h_{n n}\\
%	 &\left(\sum_{i=1}^n \xi_i h_{i\alpha} \right)- \xi_\alpha v\\
%	 & 2 \xi_n X_{n} - (\sum_{\alpha=1}^{n-1} \xi_\alpha X_{\alpha} + \xi_n X_{n}) \\
%	 & \xi_n X_\alpha + \xi_\alpha X_n\\
%	 & \omega_{ij}.
%\end{align*}

The boundary operator $P$ is said to be \emph{complementing} (see \cite[p. 42]{Agmon:1964us}) if for any $\xi= (\xi', \tau )$ with $|\xi'| = 1$ and for any constants $C_1,\dots, C_N$ satisfying
\begin{align}\label{eq:com}
	\begin{bmatrix} C_1 & \cdots & C_N\end{bmatrix}  P'(\xi) \, \adj L'(\xi) \equiv \begin{bmatrix}0\\ \vdots \\ 0 \end{bmatrix}   \; \mbox{ mod } M^+(\tau)
\end{align}
implies $C_1=\dots= C_N=0$, where  $M^+(\tau) = (\tau-\mathsf{i})^N$, and $\mathrm{i}$ is the root (with multiplicity $N$) of  $\det L'(\xi)=0$ with positive imaginary part. Since the principle symbol $L'(\xi)$is an upper triangular matrix with all diagonal entries proportional to $|\xi|^2$, it is direct to verify that the adjoint matrix $\adj L'(\xi):= \Det L'(\xi) (L'(\xi))^{-1} $ has the form $|\xi|^{2N-1}$ times a nondegenerate upper   triangular matrix  $D$ with identity entries on the diagonal. Now, taking $\tau = \mathsf i$ in \eqref{eq:com} gives the homogeneous system 
\[
	\begin{bmatrix} C_1 & \cdots & C_N\end{bmatrix}  P'(\xi', \mathsf i) D \equiv \begin{bmatrix}0\\ \vdots \\ 0 \end{bmatrix}   
\]
Since we have verify that $P'(\xi', \mathsf i)$ is nondegenerate by \eqref{eq:bd0}, we conclude that $C_1=\dots=C_N=0$. 
\end{proof}

\begin{proof}[{Proof of Proposition~\ref{pr:sur}}]
It suffices to prove \eqref{it:bd}. We have shown that $D\Phi^\varphi_{(g, k)} : \mathcal S^{m,\alpha}  \to \C^{m-2,\alpha}_{-2-q}(M)$ has closed range. To show that it is surjective, it suffices to prove that the adjoint operator has trivial kernel. That is, for any $(f, X)\in (C^{m-2,\alpha}_{-2-q}(M))^*$ satisfying $(D\Phi^\varphi_{(g, k)})^*(f, X)=0$, we must have $(f, X)\equiv 0$. By elliptic regularity, $(f, X)$ is at least $\C^{2,\alpha}_{\mathrm{loc}}(M)$. Then, by the expansion in Lemma~\ref{le:asy}, and the fact that $(f, X)$ are bounded linear functionals, we conclude $(f, X)\equiv 0$. The surjectivity of $D\Phi^\varphi_{(g, k)}$ on the space $\mathcal S^{m,\alpha}$ with homogeneous boundary data  readily implies the surjectivity of the operator 
 $DT$, and hence $T$ is locally surjective.  
\end{proof}

\bibliographystyle{plain}
\bibliography{literature.bib}

\begin{thebibliography}{10}

\bibitem{Agmon:1964us}
S.~Agmon, A.~Douglis, and L.~Nirenberg.
\newblock Estimates near the boundary for solutions of elliptic partial differential equations satisfying general boundary conditions. {II}.
\newblock {\em Comm. Pure Appl. Math.}, 17:35--92, 1964.

\bibitem{Agostiniani-Mantegazza-Mazzieri-Oronzio:2025}
Virginia Agostiniani, Carlo Mantegazza, Lorenzo Mazzieri, and Francesca Oronzio.
\newblock Riemannian {P}enrose inequality via nonlinear potential theory.
\newblock {\em arXiv:2205.11642 [math.DG]}, 2024.

\bibitem{An:2020}
Zhongshan An.
\newblock Ellipticity of {B}artnik boundary data for stationary vacuum spacetimes.
\newblock {\em Comm. Math. Phys.}, 373(3):859--906, 2020.

\bibitem{Anderson-Jauregui:2019}
Michael~T. Anderson and Jeffrey~L. Jauregui.
\newblock Embeddings, immersions and the {B}artnik quasi-local mass conjectures.
\newblock {\em Ann. Henri Poincar\'{e}}, 20(5):1651--1698, 2019.

\bibitem{Anderson-Khuri:2013}
Michael~T. Anderson and Marcus~A. Khuri.
\newblock On the {B}artnik extension problem for the static vacuum {E}instein equations.
\newblock {\em Classical Quantum Gravity}, 30(12):125005, 33, 2013.

\bibitem{Bartnik:1989}
Robert Bartnik.
\newblock A new definition of quasi-local mass.
\newblock In {\em Phys. Rev. Lett.}, volume~62, pages 2346 -- 2348, 1989.

\bibitem{Bartnik:1997}
Robert Bartnik.
\newblock Energy in general relativity.
\newblock In {\em Tsing {H}ua lectures on geometry \& analysis ({H}sinchu, 1990--1991)}, pages 5--27. Int. Press, Cambridge, MA, 1997.

\bibitem{Bartnik:2002}
Robert Bartnik.
\newblock Mass and 3-metrics of non-negative scalar curvature.
\newblock In {\em Proceedings of the {I}nternational {C}ongress of {M}athematicians, {V}ol. {II} ({B}eijing, 2002)}, pages 231--240. Higher Ed. Press, Beijing, 2002.

\bibitem{Bartnik:2005}
Robert Bartnik.
\newblock Phase space for the {E}instein equations.
\newblock {\em Comm. Anal. Geom.}, 13(5):845--885, 2005.

\bibitem{Beig-Chrusciel:1996}
Robert Beig and Piotr~T. Chru\'sciel.
\newblock Killing vectors in asymptotically flat space-times. {I}. {A}symptotically translational {K}illing vectors and the rigid positive energy theorem.
\newblock {\em J. Math. Phys.}, 37(4):1939--1961, 1996.

\bibitem{Bray:2001}
Hubert~L. Bray.
\newblock Proof of the {R}iemannian {P}enrose inequality using the positive mass theorem.
\newblock {\em J. Differential Geom.}, 59(2):177--267, 2001.

\bibitem{Chrusciel-Costa-Heusler:2012}
Piotr~T. Chru\'sciel, Jo\~{a}o~L. Costa, and Markus Heusler.
\newblock Stationary black holes: uniqueness and beyond.
\newblock {\em Living Rev. Relativ.}, 15(7), 2012.

\bibitem{Chrusciel-Costa:2008}
Piotr~T. Chru\'sciel and Jo\~ao~Lopes Costa.
\newblock On uniqueness of stationary vacuum black holes.
\newblock Number 321, pages 195--265. 2008.
\newblock G\'eom\'etrie diff\'erentielle, physique math\'ematique, math\'ematiques et soci\'et\'e. I.

\bibitem{Chrusciel-Delay-Galloway-Howard:2001}
Piotr~T. Chru\'sciel, Erwann Delay, Gregory~J. Galloway, and Ralph Howard.
\newblock Regularity of horizons and the area theorem.
\newblock {\em Ann. Henri Poincar\'e}, 2(1):109--178, 2001.

\bibitem{Chrusciel-Maerten:2006}
Piotr~T. Chru\'sciel and Daniel Maerten.
\newblock Killing vectors in asymptotically flat space-times. {II}. {A}symptotically translational {K}illing vectors and the rigid positive energy theorem in higher dimensions.
\newblock {\em J. Math. Phys.}, 47(2):022502, 10, 2006.

\bibitem{Chrusciel-Wald:1994}
Piotr~T. Chru\'sciel and Robert~M. Wald.
\newblock Maximal hypersurfaces in stationary asymptotically flat spacetimes.
\newblock {\em Comm. Math. Phys.}, 163(3):561--604, 1994.

\bibitem{Corvino:2000}
Justin Corvino.
\newblock Scalar curvature deformation and a gluing construction for the {E}instein constraint equations.
\newblock {\em Comm. Math. Phys.}, 214(1):137--189, 2000.

\bibitem{Corvino-Huang:2020}
Justin Corvino and Lan-Hsuan Huang.
\newblock Localized deformation for initial data sets with the dominant energy condition.
\newblock {\em Calc. Var. Partial Differential Equations}, 59(1):Paper No. 42, 43, 2020.

\bibitem{Corvino-Schoen:2006}
Justin Corvino and Richard Schoen.
\newblock On the asymptotics for the vacuum {E}instein constraint equations.
\newblock {\em J. Differential Geom.}, 73(2):185--217, 2006.

\bibitem{Eichmair:2013}
Michael Eichmair.
\newblock The {J}ang equation reduction of the spacetime positive energy theorem in dimensions less than eight.
\newblock {\em Comm. Math. Phys.}, 319(3):575--593, 2013.

\bibitem{Hill:1970}
C.~Denson Hill.
\newblock A sharp maximum principle for degenerate elliptic-parabolic equations.
\newblock {\em Indiana Univ. Math. J.}, 20:213--229, 1970/71.

\bibitem{HKK:2022}
Sven Hirsch, Demetre Kazaras, and Marcus Khuri.
\newblock Spacetime harmonic functions and the mass of 3-dimensional asymptotically flat initial data for the {E}instein equations.
\newblock {\em J. Differential Geom.}, 122(2):223--258, 2022.

\bibitem{HirschZhang:2023}
Sven Hirsch and Yiyue Zhang.
\newblock The case of equality for the spacetime positive mass theorem.
\newblock {\em J. Geom. Anal.}, 33(1):30, 2023.

\bibitem{Hirsch-Zhang:2024}
Sven Hirsch and Yiyue Zhang.
\newblock Initial data sets with vanishing mass are contained in pp-wave spacetimes.
\newblock {\em arXiv:2403.15984 [gr-qc]}, 2024.

\bibitem{Hounnonkpe-Minguzzi:2025}
Raymond~A. Hounnonkpe and Ettore Minguzzi.
\newblock Regularity and temperature of stationary black hole event horizons.
\newblock {\em Comm. Math. Phys.}, 406(9):Paper No. 228, 31, 2025.

\bibitem{Huang-Jang:2022}
Lan-Hsuan Huang and Hyun~Chul Jang.
\newblock Scalar curvature deformation and mass rigidity for {ALH} manifolds with boundary.
\newblock {\em Trans. Amer. Math. Soc.}, 375(11):8151--8191, 2022.

\bibitem{Huang-Lee:2020}
Lan-Hsuan Huang and Dan~A. Lee.
\newblock Equality in the spacetime positive mass theorem.
\newblock {\em Comm. Math. Phys.}, 376(3):2379--2407, 2020.

\bibitem{Huang-Lee:2024}
Lan-Hsuan Huang and Dan~A. Lee.
\newblock Bartnik mass minimizing initial data sets and improvability of the dominant energy scalar.
\newblock {\em J. Differential Geom.}, 126(2):741--800, 2024.

\bibitem{Huang-Lee:2025}
Lan-Hsuan Huang and Dan~A. Lee.
\newblock Equality in the spacetime positive mass theorem {II}.
\newblock {\em Calc. Var. Partial Differential Equations}, 64(3):Paper No. 92, 16, 2025.

\bibitem{Huisken-Ilmanen:2001}
Gerhard Huisken and Tom Ilmanen.
\newblock The inverse mean curvature flow and the {R}iemannian {P}enrose inequality.
\newblock {\em J. Differential Geom.}, 59(3):353--437, 2001.

\bibitem{Schoen-Yau:1979-pmt1}
Richard Schoen and Shing-Tung Yau.
\newblock On the proof of the positive mass conjecture in general relativity.
\newblock {\em Comm. Math. Phys.}, 65(1):45--76, 1979.

\bibitem{Schoen-Yau:1981-pmt2}
Richard Schoen and Shing-Tung Yau.
\newblock Proof of the positive mass theorem. {II}.
\newblock {\em Comm. Math. Phys.}, 79(2):231--260, 1981.

\bibitem{Witten:1981}
Edward Witten.
\newblock A new proof of the positive energy theorem.
\newblock {\em Comm. Math. Phys.}, 80(3):381--402, 1981.

\bibitem{Yip:1987}
P.~F. Yip.
\newblock A strictly-positive mass theorem.
\newblock {\em Comm. Math. Phys.}, 108(4):653--665, 1987.

\end{thebibliography}

\end{document}